\documentclass{sig-alt-full}
\usepackage[ruled,vlined]{algorithm2e}
\usepackage{comment}
\usepackage{graphicx}
\usepackage{url}
\usepackage[pdftex,colorlinks=true,citecolor=black,filecolor=black,%
            linkcolor=black,urlcolor=black]{hyperref}
\usepackage{listings}
\usepackage{enumitem}
\usepackage{subfig}
\usepackage[countmax]{subfloat}
\usepackage{color}
\usepackage{amsmath}
\usepackage{amsthm}

\newtheorem{define}{Definition}
\newtheorem{claim}{Claim}
\newtheorem{theorem}{Theorem}
\newtheorem{lemma}{Lemma}

\renewcommand{\paragraph}[1]{\vspace{0.5em}\noindent{\bf #1:}}

\newcommand{\equsize}{\normalsize}

\excludecomment{conference}
\includecomment{techreport}

\begin{conference}
\end{conference}

\title{Reuse It Or Lose It: More Efficient Secure Computation\\
Through Reuse of Encrypted Values}

\author{
{\normalsize Benjamin Mood} \\
{\sf\normalsize bmood@ufl.edu}\\
University of Florida
\and
{\normalsize Debayan Gupta} \\
{\sf\normalsize debayan.gupta@yale.edu}\\
Yale University
\and
{\normalsize Kevin Butler} \\
{\sf\normalsize butler@ufl.edu}\\
University of Florida
\and
{\normalsize Joan Feigenbaum} \\
{\sf\normalsize joan.feigenbaum@yale.edu}\\
Yale University
}

\begin{document}

\maketitle

\begin{abstract}
\noindent
Two-party secure function evaluation (SFE) has become significantly more
feasible, even on resource-constrained devices, because of advances in
server-aided computation systems. However, there are still bottlenecks,
particularly in the input validation stage of a computation. Moreover, SFE
research has not yet devoted sufficient attention to the important problem of
retaining state after a computation has been performed so that expensive
processing does not have to be repeated if a similar computation is done again.
This paper presents PartialGC, an SFE system that allows the reuse of encrypted
values generated during a garbled-circuit computation. We show that using
PartialGC can reduce computation time by as much as 96\% and bandwidth by as
much as 98\% in comparison with previous outsourcing schemes for secure
computation. We demonstrate the feasibility of our approach with two sets of
experiments, one in which the garbled circuit is evaluated on a mobile device
and one in which it is evaluated on a server. We also use PartialGC to build a
privacy-preserving ``friend finder'' application for Android. The reuse of
previous inputs to allow stateful evaluation represents a new way of looking at
SFE and further reduces computational barriers.
\end{abstract}

\section{Introduction}

Secure function evaluation, or {\it SFE}, allows multiple parties to jointly
compute a function while maintaining input and output privacy. The two-party
variant, known as 2P-SFE, was first introduced by Yao in the
1980s~\cite{Yao1982} and was largely a theoretical curiosity. Developments in
recent years have made 2P-SFE vastly more
efficient~\cite{Huang2011,Kreuter2012,shelat13}. However, computing a function
using SFE is still usually much slower than doing so in a non-privacy-preserving
manner. Garbled circuits, described by Yao, are a powerful mechanism for 
performing SFE, with modern variants allowing malicious security for complex 
programs.

As mobile devices become more powerful and ubiquitous, users expect more
services to be accessible through them. When SFE is performed on mobile
devices (where resource constraints are tight), it is extremely slow\ --\
{\em if} the computation can be run at all without exhausting the memory,
which can happen for non-trivial input sizes and algorithms~\cite{CMTB2013}.
One way to allow mobile devices to perform SFE is to use a server-aided
computational model~\cite{CMTB2013,Kamara2012}, allowing the majority of an
SFE computation to be ``outsourced'' to a more powerful device while still
preserving privacy. Past approaches, however, have not considered the ways
in which mobile computation differs from the desktop. Often, the mobile
device is called upon to perform {\em incremental} operations that are
continuations of a previous computation. 

Consider, for example, a
{\it friend finder} application where the location of users is updated
periodically to determine whether a contact is in proximity. Traditional
applications disclose location information to a central server. A
privacy-preserving {\it friend finder} could perform these operations in a
mutually oblivious fashion. However, every incremental location update would
require a full re-evaluation of the function with fresh inputs in a standard
SFE solution. Our examination of an outsourced SFE scheme for mobile devices
by Carter et al.~\cite{CMTB2013} (hereon CMTB), determined that the
cryptographic consistency checks performed on the inputs to an SFE
computation {\em themselves} can constitute the greatest bottleneck to
performance.

Additionally, many other applications require the ability to save state, a
feature that current garbled circuit implementations do not possess. The
ability to save state and reuse an intermediate value from one garbled
circuit execution to another would be useful in many other ways, {\it
e.g.}, we could split a large computation into a number of smaller pieces.
Combined with efficient input validation, this becomes an extremely
attractive proposition.

In this paper, we show that it is possible to reuse an encrypted value in an
outsourced SFE computation (we use a cut-and-choose garbled circuit
protocol) even if one is restricted to primitives that are part of standard
garbled circuits. Our system, PartialGC, which is based on CMTB, provides a
way to take encrypted output wire values from one SFE computation, save
them, and then reuse them as input wires in a new garbled circuit. Our
method vastly reduces the number of cryptographic operations compared to the
trivial mechanism of simply XOR'ing the results with a one-time pad, which
requires either generating inside the circuit, or inputting, a very large
one-time pad, both complex operations.  Through the use of improved input
validation mechanisms proposed by shelat and Shen~\cite{shelat13} (hereon
sS13) and new methods of {\em partial input} gate checks and evaluation, we
improve on previous proposals. There are other approaches to the creation of
reusable garbled circuits~\cite{Goldwasser2013,Gentry2013,Bran2013}, and
previous work on reusing encrypted values in the ORAM
model~\cite{LO13,GHLORW14,LO14}, but these earlier schemes have not been
implemented. By contrast, we have implemented our scheme and found it to be
both practical and efficient; we provide a performance analysis and a sample
application to illustrate its feasibility (Section~\ref{sec:experiments}),
as well as a simplified example execution (Appendix~\ref{appendix:example}).

By breaking a large program into smaller pieces, our system allows 
interactive I/O throughout the garbled circuit computation. To the best of our
knowledge this is the first practical protocol for performing interactive I/O in
the middle of a cut-and-choose garbled circuit computation. 

%

Our system comprises three parties - a generator, an evaluator, and a third
party (``the cloud''), to which the evaluator outsources its part of the
computation. Our protocol is secure against a malicious adversary, assuming
that there is no collusion by either party with the cloud. We also provide 
a semi-honest version of the protocol.

\begin{figure}[t]
\centering
\includegraphics[width=3in]{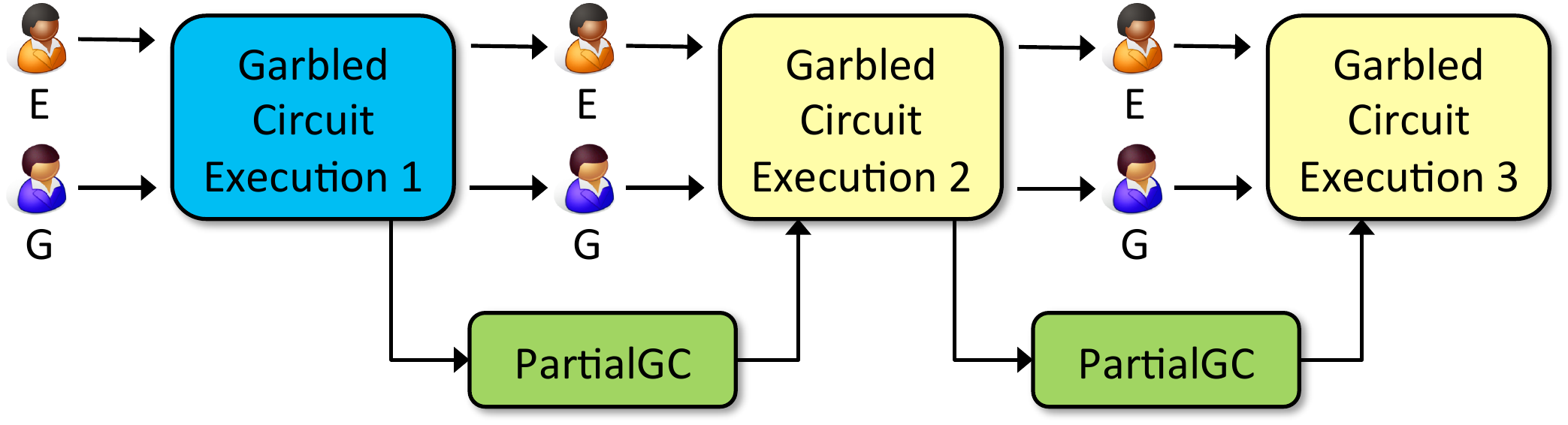}
\caption{PartialGC Overview. E is evaluator and G is generator. The blue box is a standard execution that produces partial outputs (garbled values); yellow boxes represent executions that take partial inputs and produce partial outputs.}
\label{fig:overview}
\end{figure}

\begin{figure}[t]
\centering
\includegraphics[width=3in]{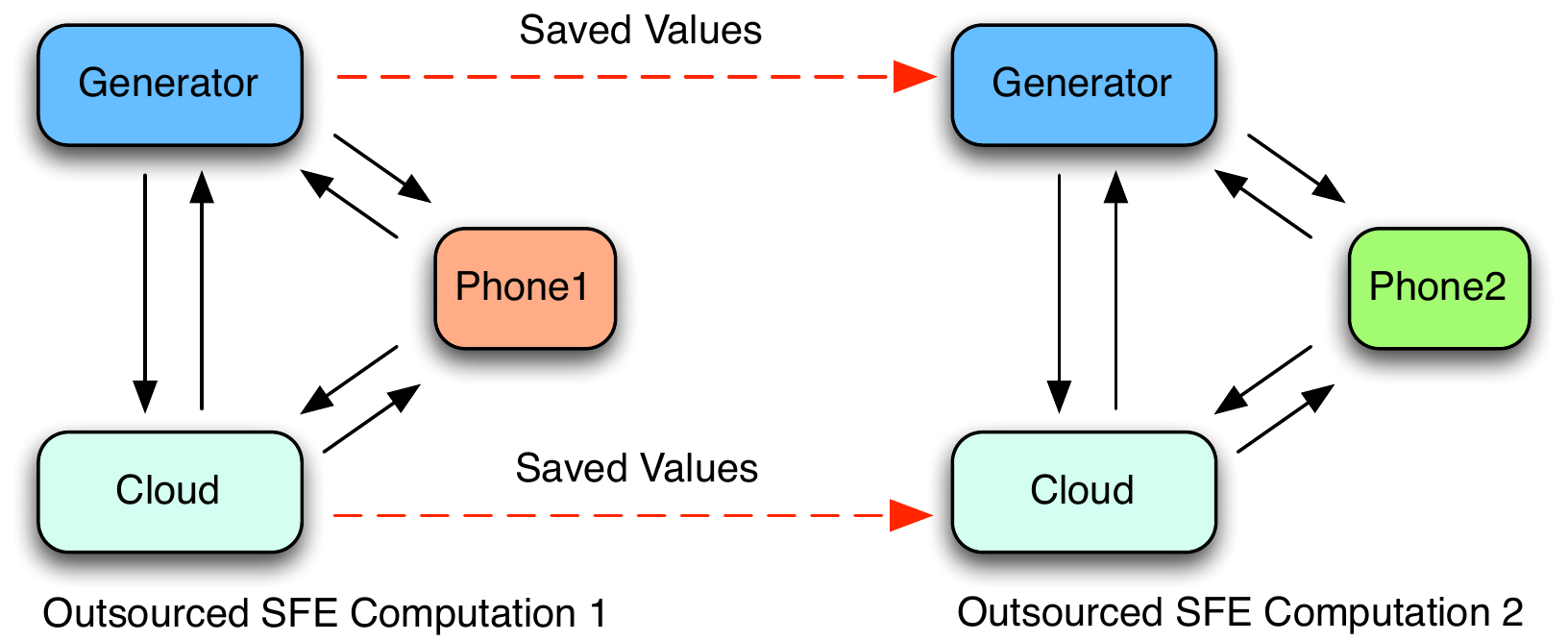}
\caption{Our system has three parties. Only the cloud and generator have to save intermediate values - this means that we can have different phones in different computations.}
\label{fig:outsource}
\end{figure}

Figure~\ref{fig:overview} shows how PartialGC works at a high level: First,
a standard SFE execution (blue) takes place, at the end of which we ``save''
some intermediate output values. All further executions use intermediate
values from previous executions.  In order to reuse these values,
information from both parties -- the generator and the evaluator -- has to
be saved. In our protocol, it is the cloud -- rather than the evaluator --
that saves information. This allows multiple distinct evaluators to
participate in a large computation over time by saving state in the cloud
between different garbled circuit executions. For example, in a scenario
where a mobile phone is outsourcing computation to a cloud, PartialGC can
save the encrypted intermediate outputs to the cloud instead of the phone
(Figure~\ref{fig:outsource}). This allows the phones to communicate with
each other by storing encrypted intermediate values in the cloud, which is
more efficient than requiring them to directly participate in the saving of
values, as required by earlier 2P-SFE systems. Our friend finder
application, built for an Android device, reflects this usage model and
allows multiple friends to share their intermediate values in a cloud.
Other friends use these saved values to check whether or not someone is in
the same map cell as themselves without having to copy and send data.


By incorporating our optimizations, we give the following contributions:

\setlist{nolistsep=3mm,leftmargin=4mm}
\begin{enumerate}

\item {\it Reusable Encrypted Values} -- We show how to reuse an encrypted value, using only garbled circuits, by mapping one garbled value into another.

\item {\it Reduced Runtime and Bandwidth} -- We show how reusable encrypted values can be used in practice to reduce the execution time for a garbled-circuit computation; we get a 96\% reduction in runtime and a 98\% reduction in bandwidth over CMTB. 
\begin{techreport}
Impressively, we can reduce the amount of bandwidth required by the mobile party {\it arbitrarily} when no input checks have to be performed on the partial (intermediate) inputs in our protocol.
\end{techreport}

\item {\it Outsourcing Stateful Applications} -- We show how our system increases the scope of SFE applications by allowing multiple evaluating parties over a period of time to operate on the saved state of an SFE computation without the need for these parties to know about each other. 

\end{enumerate}

\noindent
The remainder of our paper is organized as follows:
Section~\ref{sec:background} provides some background on SFE.
Section~\ref{sec:pgc} introduces the concept of partial garbled circuits in detail. The
PartialGC protocol and its implementation are described in
Section~\ref{sec:protocolchanges}, while its security is analyzed in 
Section~\ref{sec:sec}.
Section~\ref{sec:experiments} evaluates PartialGC and introduces the friend
finder application. Section~\ref{sec:relwork}
discusses related work and Section~\ref{sec:conc} concludes.

\section{Background}
\label{sec:background}

Secure function evaluation (SFE) addresses scenarios where two or more
mutually distrustful parties $P_1,\dots, P_n$, with private inputs
$x_1,\dots, x_n$, want to compute a given function $y_i = f(x_1,\dots, x_n)$
($y_i$ is the output received by $P_i$), such that no $P_i$ learns anything
about any $x_j$ or $y_j$, $i \neq j$ that is not logically implied by $x_i$
and $y_i$. Moreover, there exists no trusted third party -- if there was, the
$P_i$s could simply send their inputs to the trusted party, which would
evaluate the function and return the $y_i$s.

SFE was first proposed in the 1980s in Yao's seminal paper~\cite{Yao1982}.
The area has been studied extensively by the cryptography community, leading
to the creation of the first general purpose platform for SFE,
Fairplay~\cite{Malkhi2004} in the early 2000s. Today, there exist many such
platforms
~\cite{Burkhart2010,Damgard2009,Henecka2010,Holzer2012,Kreuter2013,Shelat2011,Zhang2013}.

The classic platforms for 2P-SFE, including Fairplay, use garbled circuits.
A garbled circuit is a Boolean circuit which is encrypted in such a way that it
can be evaluated when the proper input wires are entered. The party that
evaluates this circuit does not learn anything about what any particular wire
represents. In 2P-SFE, the two parties are: the {\it generator}, which creates
the garbled circuit, and the {\it evaluator}, which evaluates the garbled
circuit. Additional cryptographic techniques are used for input and output; we
discuss these later.

A two-input Boolean gate has four truth table entries. A two-input garbled
gate also has a truth table with four entries representing 1s and 0s, but these
entries are encrypted and can only be retrieved when the proper keys are used.
The values that represent the 1s and 0s are random strings of bits. The truth
table entries are permuted such that the evaluator cannot determine which entry
she is able to decrypt, only that she is able to decrypt an entry. The entirety
of a garbled gate is the four encrypted output values.

Each garbled gate is then encrypted in the following way: Each entry in the
truth table is encrypted under the two input wires, which leads to the result,
$truth_i=Enc(input_x | |$ $ input_y)\oplus output_i$, where $truth_i$ is a value in
the truth table, $input_x$ is the value of input wire $x$, $input_y$ is the
value of input wire $y$, and $output_i$ is the non-encrypted value, which
represents either 0 or 1.We use AES as the $Enc$ function. If the evaluator has
$input_x$ and $input_y$, then she can also receive $output_i$, and the encrypted
truth tables are sent to her for evaluation.

For the evaluator's input, 1-out-of-2 oblivious transfers
(OTs)~\cite{Bellare1990,Ishai2003,Naor1999a,Naor2001} are
used. In a 1-out-of-2 OT, one party offers up two possible
values while the other party selects one of the two values
without learning the other. The party that offers up the
two values does not learn which value was selected. Using this
technique, the evaluator gets the wire labels for her input
without leaking information.

The only way for the evaluator to get a correct output value from a garbled
gate is to know the correct decryption keys for a specific entry in the
truth table, as well as the location of the value she has to decrypt.

During the permutation stage, rather than simply randomly permuting the values,
the generator permutes values based on a specific bit in $input_x$ and
$input_y$, such that, given $input_x$ and $input_y$ the evaluator knows that the
location of the entry to decrypt is $bit_x*2+bit_y$. These bits are called the
{\it permutation bits}, as they show the evaluator which entry to select based
on the permutation; this optimization, which does not leak any information, is
known as {\em point and permute}~\cite{Malkhi2004}.

\subsection{Threat Models}

Traditionally, there are two threat models discussed in SFE work,
semi-honest and malicious. The above description of garbled circuits is the
same in both threat models. In the semi-honest model users stay true to the
protocol but may attempt to learn extra information from the system by
looking at any message that is sent or received. In the
malicious model, users may attempt to change anything with the goal of
learning extra information or giving incorrect results without being
detected; extra techniques must be added to achieve
 security against a malicious adversary.

There are several well-known attacks a malicious adversary could use against
a garbled circuit protocol. A protocol secure against malicious adversaries
must have solutions to all potential pitfalls, described in turn:

{\it Generation of incorrect circuits}:
If the generator does not create a correct garbled circuit, he could learn extra
information by modifying truth table values to output the evaluator's input; he
is limited only by the external structure of the garbled circuit the evaluator
expects.

{\it Selective failure of input}:
If the generator does not offer up correct input wires to the evaluator, and
the evaluator selects the wire that was not created properly, the generator
can learn up to a single bit of information based on whether the computation
produced correct outputs.

{\it Input consistency}:
If either party's input is not consistent across all circuits, then it might
be possible for extra information to be retrieved. 

{\it Output consistency}:
In the two-party case, the output consistency check verifies that the evaluator
did not modify the generator's output before sending it.

\subsubsection{Non-collusion}\label{sec:background:noncollusion}
CMTB assumes non-collusion, as quoted below:

\smallskip

\noindent{\it ``The outsourced two-party SFE protocol securely computes a function f(a,b) in the following two corruption scenarios: 
(1)The cloud is malicious and non-cooperative with respect to the rest of the parties, while all other parties are semi-honest, (2)All but one party is malicious, while the cloud is semi-honest.''}

\smallskip

This is the standard definition of non-collusion used in server-aided works such as Kamara et al.~\cite{Kamara2012}. Non-collusion does not mean the parties are trusted; it only means the two parties are not working together in order to cheat. In CMTB, any individual party that attempts to cheat to gain additional information will still be caught, but collusion between multiple parties could leak information. For instance, the generator could send the cloud the keys to decrypt the circuit and see what the intermediate values are of the garbled function.

\section{Partial Garbled Circuits}
\label{sec:pgc}

We introduce the concept of {\em partial garbled circuits} (PGCs), which allows
the encrypted wire outputs from one SFE computation to be used as inputs to
another. This can be accomplished by {\em mapping} the encrypted output wire
values to valid input wire values in the next computation. In order to better
demonstrate their structure and use, we first present PGCs in a semi-honest
setting, before showing how they can aid us against malicious adversaries.

\subsection{PGCs in the Semi-Honest Model}

In the semi-honest model, for each wire value, the generator can simply send two
values to the evaluator, which transforms the wire label the evaluator owns to
work in another garbled circuit. Depending on the point and permute bit of the
wire label received by the evaluator, she can map the value from a previous
garbled circuit computation to a valid wire label in the next computation. 

Specifically, for a given wire pair, the generator has wires $w_0^{t-1}$ and
$w_1^{t-1}$, and creates wires $w_0^{t}$ and $w_1^{t}$. Here, $t$ refers to a
particular computation in a series, while 0 and 1 correspond to the values of
the point and permute bits of the $t-1$ values. The generator sends the values
$w_0^{t-1}\oplus w_0^{t}$ and $w_1^{t-1}\oplus w_1^{t}$ to the evaluator.
Depending on the point and permute bit of the $w_i^{t-1}$ value she possesses,
the evaluator selects the correct value and then XORs her $w_i^{t-1}$ with the
($w_i^{t-1}\oplus w_i^{t}$) value, thereby giving her $w_i^{t}$, the valid
partial input wire.

\subsection{PGCs in the Malicious Model}

In the malicious model we must allow the evaluation of a circuit with
partial inputs and verification of the mappings, while preventing a selective
failure attack. The following features are necessary to accomplish these goals:

\smallskip
\noindent{\it Verifiable Mapping}:
The generator $G$ is able to create a secure mapping from a saved garbled wire
value into a new computation that can be checked by the evaluator $E$, without
$E$ being able to reverse the mapping. During the evaluation and check phase, $E$ must be able to
verify the mapping $G$ sent. $G$ must have either committed to the mappings
before deciding the partition of evaluation and check circuits, or never learned
which circuits are in the check versus the evaluation sets.

\begin{figure}[t]
\centering
\includegraphics[width=2in]{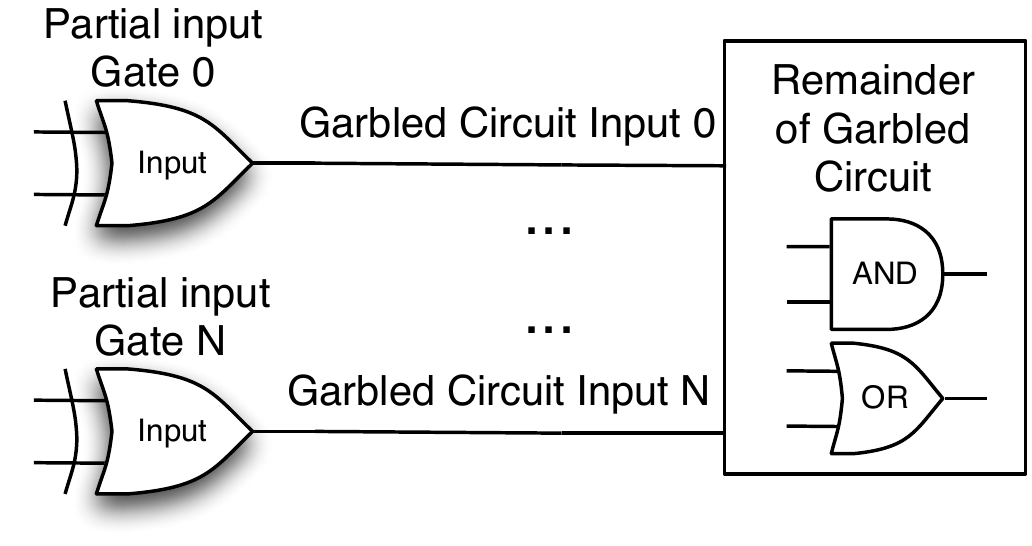}
\caption{This figure shows how we create a single {\it partial input gate} for each input bit for each circuit and then link the {\it partial input gates} to the remainder of the circuit.}
\label{fig:pinputs}
\end{figure}

\smallskip
\noindent{\it Partial Generation and Partial Evaluation}:
$G$ creates the garbled gates necessary for $E$ to enter the previously output
intermediate encrypted values into the next garbled circuit. These garbled gates
are called {\it partial input gates}. As shown in Figure \ref{fig:pinputs} each
garbled circuit is made up of two pieces: the partial input gates and the
remainder of the garbled circuit.

\smallskip
\noindent{\it Revealing Incorrect Transformations}:
Our last goal is to let $E$ inform $G$ that incorrect values have been
detected. Without a way to limit leakage, $G$ could gain information based on
whether or not $E$ informs $G$ that she caught him cheating. This is a selective
failure attack and is not present in our protocol.

\section{PartialGC Protocol} \label{sec:protocolchanges}
We start with the CMTB protocol and add cut-and-choose operations from sS13
before introducing the mechanisms needed to save and reuse values. We defer to
the original papers for full details of the outsourced oblivious
transfer~\cite{CMTB2013} and the generator's input consistency
check~\cite{shelat13} sub-protocols that we use as primitives in our protocol.

Our system operates in the same threat model as CMTB (see
Section~\ref{sec:background:noncollusion}): we are secure against a malicious
adversary under the assumption of non-collusion. A description of the CMTB
protocol is available in Appendix~\ref{appendix:CMTB}.


\subsection{Preliminaries} 


There are three participants in the protocol: 

{\it Generator} -- The generator is the party that generates the
garbled circuit for the 2P-SFE. 

{\it Evaluator} --The evaluator is the other party in the 2P-SFE; it outsources computation to the cloud. 

{\it Cloud} -- The cloud is the third party that executes the garbled
circuit outsourced by the evaluator.

\smallskip\smallskip\noindent{\bf Notation}


\smallskip\noindent {\it $C_i$} - The $i$th circuit.

\smallskip\noindent {\it $CKey_i$} - Circuit key used for the free XOR
optimization~\cite{Kolesnikov2008}. The key is randomly generated and then used
as the difference between the 0 and 1 wire labels for a circuit $C_i$.

\smallskip\noindent {\it $CSeed_i$} - This value is created by the generator's
PRNG and is used to generate a particular circuit $C_i$.

\smallskip\noindent {\it $POut\#_{i,j}$} - The {\it partial output} values are the encrypted wire values
output from an SFE computation. These are encrypted garbled circuit values that
can be reused in another garbled circuit computation. \# is replaced in our
protocol description with either a 0, 1, or x, signifying whether it represents
a 0, 1, or an unknown value (from the cloud's point of view). $i$ denotes the
circuit the $POut$ value came from and $j$ denotes the wire of the $POut_i$
circuit.


\smallskip\noindent {\it $PIn\#_{i,j}$} - The {\it partial input} values are the re-entered
{\it POut} values after they have been obfuscated to remove the circuit key from
the previous computation. These values are input to the {\it partial input
gates}. \#, $i$, and $j$, are the same as above.

\smallskip\noindent {\it $GIn\#_{i,j}$} - The {\it garbled circuit input} values
are the results of the partial input gates and are input into the remaining
garbled circuit, as shown in Figure~\ref{fig:pinputs}. \#, $i$, and $j$, are
the same as above.

\smallskip\noindent {\it Partial Input Gates} - These are garbled gates that
take in $PIn$ values and output $GIn$ values. Their purpose is to transform the
{\it PIn} values into values that are under $CKey_i$ for the current circuit.

\begin{algorithm}[h]

\scriptsize

\DontPrintSemicolon
 \SetKwData{Left}{left}\SetKwData{This}{this}\SetKwData{Up}{up}
  \SetKwFunction{algCutandChoose}{Cut\_and\_Choose}
\SetKwFunction{algEvaluatorInput}{Evaluator\_Input}
\SetKwFunction{algGeneratorInputCheck}{Generator\_Input\_Check}
\SetKwFunction{algPartialInput}{Partial\_Input}
\SetKwFunction{algCircuitExecution}{Circuit\_Execution}
\SetKwFunction{algCircuitOutput}{Circuit\_Output}
\SetKwFunction{algPartialOutput}{Partial\_Output}
  \SetKwInOut{Input}{Input}\SetKwInOut{Output}{Output}

\setcounter{algocf}{-1}
\Input{Circuit\_File, Bit\_Security, Number\_of\_Circuits, Inputs, Is\_First\_Execution}
\Output{Circuit File Output}
\caption{PartialComputation}

\algCutandChoose{is\_First\_Execution}\;

Eval\_Garbled\_Input $\leftarrow$ \algEvaluatorInput{Eval\_Select\_Bits, Possible\_Eval\_Input}\;

\algGeneratorInputCheck{Gen\_Input}\;

Partial\_Garbled\_Input  $\leftarrow$ \algPartialInput{Partial\_$Output_{time-1}$}\;

Garbled\_Output, Partial\_Output $\leftarrow$ \algCircuitExecution{Garbled\_Input (Gen, Eval, Partial) }\;

\algCircuitOutput{Garbled\_Output}\;

\algPartialOutput{Partial\_Output}\;
\end{algorithm}

\subsection{Protocol}

Each computation is self-contained; other than what is explicitly described as
saved in the protocol, each value or property is only used for a single part of
the computation ({\it i.e.} randomness is different across computations).

\smallskip\noindent{\it Common Inputs:} The program circuit file, the bit level
security $K$, the circuit level security (number of circuits) $S$, and encryption
and commitment functions.
 
\smallskip\noindent{\it Private Inputs:} The evaluator's input $evlInput$ and
generator's input $genInput$.

\smallskip\noindent{\it Outputs:} The evaluator and generator can both receive
garbled circuit outputs.

\smallskip\noindent
{\bf Phase 1: Preparation and Cut-and-choose}

\begin{algorithm}[h]

\scriptsize

\DontPrintSemicolon
 \SetKwData{Left}{left}\SetKwData{This}{this}\SetKwData{Up}{up}
  \SetKwFunction{algCutandChoose}{Cut\_and\_Choose}
\SetKwFunction{algEvaluatorInput}{Evaluator\_Input}
\SetKwFunction{algGeneratorInputCheck}{Generator\_Input\_Check}
\SetKwFunction{algPartialInput}{Partial\_Input}
\SetKwFunction{algCircuitExecution}{Circuit\_Execution}
\SetKwFunction{algCircuitOutput}{Circuit\_Output}
\SetKwFunction{algPartialOutput}{Partial\_Output}
  \SetKwInOut{Input}{Input}\SetKwInOut{Output}{Output}
\SetKwData{Left}{left}\SetKwData{This}{this}\SetKwData{Up}{up}
\Input{is\_First\_Execution}
\caption{Cut\_and\_Choose}

 \If{is\_First\_Execution}{
circuitSelection $\leftarrow$ rand() // bit-vector of size $S$ \;
}
$N \leftarrow \frac{2}{5}S$ // Number of evaluation circuits \;

//Generator creates his garbled input and circuit seeds for each
circuit \;
\For{$i\leftarrow 0$ \KwTo $S$}{
	$CSeed_i \leftarrow rand()$\;
	$garbledGenInput_i \leftarrow garble(genInput, rand())$\;
	//generator creates or loads keys\;
	 \eIf{is\_First\_Execution}{
	$checkKey_i \leftarrow rand()$\;
	$evlKey_i \leftarrow rand()$ \;
	}
	{
		loadKeys();\;
		$checkKey_i \leftarrow hash(loadedCheckKey_i) $\;
		$evlKey_i \leftarrow hash(loadedEvlKey_i) $\;
	}
	// encrypts using unique one-time XOR pads\;
	$encSeedIn_i \leftarrow CSeed_i \oplus evlKey_i$ \;
	$encGarbledIn_i \leftarrow garbledGenInput_i \oplus checkKey_i$ \;
}

 \eIf{is\_First\_Execution}{
// generator offers input OR keys for each circuit seed\;
$selectedKeys \leftarrow OT(circuitSelection,\{evlKey,checkKey\})$\;
}
{
loadSelectedKeys()\;
}

\For{$i\leftarrow 0$ \KwTo $S$}{
	$ genSendToEval(hash(checkKey_i),$ $hash(evaluationKey_i))$\;
}

\For{$i\leftarrow 0$ \KwTo $S$}{
	$cloudSendToEval(hash(selectedKey_i),$ $isCheckCircuit_i)$\;
}

// If all values match, the evaluator learns split, else abort. \;
\For{$i\leftarrow 0$ \KwTo $S$}{
	$ j \leftarrow isCheckCircuit_i$\;
	$correct \leftarrow ( recievedGen_{i,j} == recievedEvl_i ) $\;
	\If{!correct}{
		abort()\;
	}
}
\end{algorithm}

\noindent{\it Preparation}:

The generator creates two seeds for each circuit $C_0 \dots C_{S-1}$, $CSeed_i = \{0,1\}^K$ .

We prepare our circuits such that any output to the generator or evaluator
is output under a one-time pad, encrypted inside of the circuit. That is we augment all circuits such that $out_{evl} = out_{evl} \oplus outputKey_{evl}$ and $out_{gen} = out_{gen} \oplus outKey_{gen}$, where $out_{evl}$ and $out_{gen}$ is the initial output.

The generator and evaluator's input is extended to include the corresponding $outputKey$ and  a $K$-bit secret key for a MAC.

Using the same technique as CMTB for input encoding to split the evaluator's input in $K$ bits, where $bit_{j,0} \oplus \dots \oplus bit_{j,K-1} = evlInput_j$ for the $j$th bit of the evaluator's input. The generator then creates the possible evaluator's input for each circuit $C_i$. To create the evaluator's input, the generator creates a key $IKey_i = \{0,1\}^K$ for the $i$th circuit, and a set of seeds, $evlInputSeeds0_{j} = \{0,1\}^K$ and $evlInputSeeds1_{j} = \{0,1\}^K$,  where for $0 <= j < len(evlInput)$. Two seeds are created for each bit, representing 0 and 1. The garbled input values are then created:

\vspace{-4mm}

\equsize \begin{eqnarray*}
garbledInputEvl0_{ij} = hash(evlInputSeeds0_{j}, IKey_{i})\\
garbledInputEvl1_{ij} = hash(evlInputSeeds1_{j}, IKey_{i})
\end{eqnarray*} \normalsize

As with CMTB, the possible evaluator's inputs are permuted for each different circuit to prevent the cloud from understanding what the evaluator's input maps to. The generator commits to each input value so the cloud will be able to verify he did not swap values.

\smallskip\noindent{\it Cut-and-choose}:

Unlike some other GC protocols we do not commit to the various circuits before we execute the cut-and-choose. We modify the cut-and-choose mechanism described in sS13 as we have an extra
party involved in the computation. In this cut-and-choose, the cloud selects
which circuits are evaluation circuits and which circuits are check circuits,

\vspace{-4mm}

\equsize \begin{eqnarray*}
circuitSelection = \{0,1\}^S
\end{eqnarray*} \normalsize


\noindent  
where 0 is an evaluation circuit and 1 is a check circuit. $N$ evaluation circuits and $S-N$
check circuits are selected (like sS13, we use $N=\frac{2}{5}S$). The generator does not learn the circuit selection.

The generator generates garbled versions of his input and circuit seeds for each
circuit. He encrypts these values using unique one-time XOR pad
key for each circuit. He also encrypts the evaluator's possible input. For $0 \leq i < S$,

\vspace{-4mm}

\equsize\begin{eqnarray*}
 garbledGenInput_i = garbleInput(genInput)\\
checkKey_i =  \{0,1\}^K \\ evlKey_i =  \{0,1\}^K\\
encGarbledIn_i = garbledGenInput_i \oplus evlKey_i \\
encSeedIn_i = CSeed_i \oplus checkKey_i\\
encInputEvl = garbledInputEvl \oplus checkKey_i
\end{eqnarray*}
\normalsize

\noindent
where $garbleInput()$ takes in the input, and produces a vector of $\{0,1\}^K$ bit strings, one for each bit of the generator's input for a given $C_i$ and $garbledInputEvl$ is the garbled input\\ $(garbledInputEvl0_{i,0} || \dots || garbledInputEvl0_{i,len-1} $\\$|| garbledInputEvl1_{i,0} $ $ || \dots || garbledInputEvl1_{i,len-1})$ \\ and $len$ is the length of $evlInput$.


The cloud and generator perform an oblivious transfer where the generator offers
up decryption keys for his input and decryption keys for the circuit seed and possible evaluator's input for
each circuit. The cloud can select the key to decrypt the generator's input or
the key to decrypt the circuit seed and possible evaluator's input for a circuit but not both. 

\vspace{-4mm}

\equsize 
\begin{eqnarray*}
\quad selectedKeys = OT(circuitSelection,\{evlKey,checkKey\})
\end{eqnarray*} 
\normalsize


\noindent
For each circuit, if the cloud selects the decryption key for the circuit seed and possible evaluator's input in the
oblivious transfer, then the circuit is used as a check circuit. If the cloud selects the key for the generator's input then the circuit is used as an evaluation circuit. 

The generator sends the encrypted garbled inputs and check circuit
information for all circuits to the cloud. The cloud decrypts the information he
can decrypt using its keys. Both the cloud and generator save the decryption keys so they can be used in future computations, which use saved values.

\smallskip\noindent The evaluator must also learn the circuit split. The
generator sends a hash of each possible encryption key the cloud could have
selected to the evaluator for each circuit as an ordered pair. For $0 \leq i < S$,

\vspace{-4mm}

\equsize \begin{eqnarray*}
genSend(hash(checkKey_i), hash(evaluationKey_i))
\end{eqnarray*} \normalsize


\noindent
The cloud sends a hash of the value received to the evaluator for each circuit.
The cloud also sends bits to indicate which circuits were selected as check or
evaluation circuits to the evaluator. For $0 \leq i < S$,

\vspace{-4mm}

\equsize \begin{eqnarray*}
cloudSend(hash(selectedKey_i),isCheckCircuit_i)
\end{eqnarray*} \normalsize


\noindent
The evaluator compares the hash the cloud sent to one of the hashes the
generator sent, which is selected by the circuit selection sent by the cloud. For $0 \leq i < S$,

\vspace{-4mm}

\equsize \begin{eqnarray*}
j = isCheckCircuit_i \\
correct = (receivedGen_{i,j} ==receivedEvl_i )
\end{eqnarray*} \normalsize


\noindent
If all values match, the evaluator uses the
$isCheckCircuit_i$ to learn the split between check and evaluator circuits. Otherwise, abort.

We only perform the cut-and-choose oblivious transfer for the initial
computation. For any subsequent computations, the generator and evaluator hash
the saved decryption keys and use those hashes as the new encryption and
decryption keys. The circuit split selected by the cloud is saved and stays the
same across computations.

At the conclusion of this step (1) the cloud has all the information to evaluate the evaluation circuits when they are sent by the generator, i.e. the generator's input for each evaluation circuit, (2) the cloud has all the information to validate the check circuits when the generator sends those over, {\it i.e.}, each circuit seed and the possible evaluator's input for the check circuits (3) the cloud and evaluator know the check and evaluation circuit split, (4) the generator does not know the circuit split.

\smallskip\noindent
{\bf Phase 2: Evaluator's Input and Oblivious Transfer}

\begin{algorithm}[h]

\scriptsize

\DontPrintSemicolon
 \SetKwData{Left}{left}\SetKwData{This}{this}\SetKwData{Up}{up}
  \SetKwFunction{algCutandChoose}{Cut\_and\_Choose}
\SetKwFunction{algEvaluatorInput}{Evaluator\_Input}
\SetKwFunction{algGeneratorInputCheck}{Generator\_Input\_Check}
\SetKwFunction{algPartialInput}{Partial\_Input}
\SetKwFunction{algCircuitExecution}{Circuit\_Execution}
\SetKwFunction{algCircuitOutput}{Circuit\_Output}
\SetKwFunction{algPartialOutput}{Partial\_Output}
  \SetKwInOut{Input}{Input}\SetKwInOut{Output}{Output}
\SetKwData{Left}{left}\SetKwData{This}{this}\SetKwData{Up}{up}
\Input{Eval\_Select\_Bits, Possible\_Eval\_Input}
\Output{Eval\_Garbled\_Input }
\caption{Evaluator\_Input}

// cloud gets selected input wires
// generator offers both possible input wire values for each input wire; evaluator selects its input \;
$outSeeds = BaseOOT(bitsEvl, possibleInputs)$.  \;

// the generator sends unique IKey values for each circuit to the evaluator\;
\For{$i\leftarrow 0$ \KwTo $S$}{
	$genSendToEval( IKey_i )$\;
}

// the evaluator sends IKey values for all evaluation circuits to the cloud \;
\For{$i\leftarrow 0$ \KwTo $S$}{
	\If{!isCheckCircuit(i)}{
		$EvalSendToCloud( IKey_i )$\;
	}
}

// cloud uses this to learn appropriate inputs\;
\For{$i\leftarrow 0$ \KwTo $S$}{
\For{$j\leftarrow 0$ \KwTo $len(evlInputs)$}{
	\If{!isCheckCircuit(i)}{
		$inputEvl_{ij} \leftarrow hash(IKeys_{i}, outSeeds_{j})$\;
	}
}
}

return inputEvl\;
\end{algorithm}

We use the base outsourced oblivious transfer (OOT) of CMTB. In CMTB's OOT, the evaluator enters in the inputs buts and the generator enters in both possible inputs. The evaluator and generator perform a single OT operation before extending it, using the Ishai OT extension, to all the input bits.

After extending it across each input bit, it is then extended across each garbled circuit using the same technique described in the algorithm. After the OOT is finished, the cloud has the selected input wire values, which represent the evaluator's input. 

As with CMTB, which uses the results from a single OOT as seeds to create the evaluator's input for all circuits, the cloud in our system also uses seeds from a single base OT  (called ``BaseOOT'' below) to generate the input for the evaluation circuits. The cloud receives the seeds for each input bit selected by the evaluator.

\vspace{-4mm}

\equsize \begin{eqnarray*}
outSeeds = BaseOOT(evlInput, evlInputSeeds).  
\end{eqnarray*} \normalsize


\noindent
where $outSeeds$ are the seeds selected by the evaluator's input.

The generator sends the $IKey_i$ keys (from phase 1) to the evaluator for each circuit. The evaluator sends the keys for the evaluation circuits to the cloud. The cloud then uses these keys and the $outSeeds$ to attain the evaluator's input. For $0 \leq i < S$, for $0 \leq j < len(evlInputs)$ where $!isCheckCircuit(i)$,

\vspace{-4mm}

\equsize \begin{eqnarray*}
inputEvl_{ij} = hash(IKey_{i},outSeeds_{j})
\end{eqnarray*} \normalsize


\smallskip\noindent
{\bf Phase 3: Generator's Input Consistency Check}

\begin{algorithm}[h]

\scriptsize

\DontPrintSemicolon
 \SetKwData{Left}{left}\SetKwData{This}{this}\SetKwData{Up}{up}
  \SetKwFunction{algCutandChoose}{Cut\_and\_Choose}
\SetKwFunction{algEvaluatorInput}{Evaluator\_Input}
\SetKwFunction{algGeneratorInputCheck}{Generator\_Input\_Check}
\SetKwFunction{algPartialInput}{Partial\_Input}
\SetKwFunction{algCircuitExecution}{Circuit\_Execution}
\SetKwFunction{algCircuitOutput}{Circuit\_Output}
\SetKwFunction{algPartialOutput}{Partial\_Output}
  \SetKwInOut{Input}{Input}\SetKwInOut{Output}{Output}
\SetKwData{Left}{left}\SetKwData{This}{this}\SetKwData{Up}{up}
\Input{Generator\_Input}
\caption{Generator\_Input\_Check}

// The cloud takes a hash of the generator's input or each evaluation circuit
\For{$i\leftarrow 0$ \KwTo $S$}{
	\If{isCheckCircuit(i)}{
		$t_i \leftarrow UHF(garbledGenInput_i)$\;
	}
}

//If a single hash is different then the cloud knows the generator tried to cheat.
$correct \leftarrow ((t_0==t_1) \&(t_0 == t_2) \& \dots \& (t_0 == t_{N-1}))$\;
\If{!correct}{
	abort()\;
}
\end{algorithm}

We use the input consistency check of sS13. In this check, a universal hash is used to prove consistency of the generator's input across each evaluation circuit (attained in phase 1). Simply put, if any hash is different in any of the evaluation circuits, we know the generator did not enter consistent input. More formally, a hash of the generator's input is taken for each circuit.  For $0 < i < S$ where $!isCheckCircuit(i)$,

\vspace{-4mm}

\equsize \begin{eqnarray*}
t_i = UHF(garbledGenInput_i, C_i)
\end{eqnarray*} \normalsize

\vspace{-1mm}

\noindent
The results of these universal hashes are compared. If a single hash is different then the cloud knows the generator tried to cheat.

\vspace{-4mm}

\equsize \begin{eqnarray*}
correct =  ((t_0==t_1) \&(t_0 == t_2) \& \dots \& (t_0 == t_{N-1}))
\end{eqnarray*} \normalsize


\smallskip\noindent
{\bf Phase 4: Partial Input Gate Generation, Check, and Evaluation}

\begin{algorithm}[h]

\scriptsize

\DontPrintSemicolon
 \SetKwData{Left}{left}\SetKwData{This}{this}\SetKwData{Up}{up}
  \SetKwFunction{algCutandChoose}{Cut\_and\_Choose}
\SetKwFunction{algEvaluatorInput}{Evaluator\_Input}
\SetKwFunction{algGeneratorInputCheck}{Generator\_Input\_Check}
\SetKwFunction{algPartialInput}{Partial\_Input}
\SetKwFunction{algCircuitExecution}{Circuit\_Execution}
\SetKwFunction{algCircuitOutput}{Circuit\_Output}
\SetKwFunction{algPartialOutput}{Partial\_Output}
  \SetKwInOut{Input}{Input}\SetKwInOut{Output}{Output}
\SetKwData{Left}{left}\SetKwData{This}{this}\SetKwData{Up}{up}
\Input{Partial\_Output}
\Output{Partial\_Garbled\_Input}
\caption{Partial\_Input}

// Generation: the generator creates a {\it partial input gate}, which transforms a wire's saved values, $POut0_{i,j}$ and $POut1_{i,j}$, into values that can be used in the current garbled circuit execution, $GIn0_{i,j}$ and $GIn1_{i,j}$.\;

\For{$i\leftarrow 0$ \KwTo $S$}{
	$R_i \leftarrow PRNG.random()$\;
	\For{$j\leftarrow 0$ \KwTo $len(savedWires)$}{
		$t0 \leftarrow hash(POut0_{i,j} \oplus R_i)$\;
		$t1 \leftarrow hash(POut1_{i,j}\oplus R_i)$\;
		$PIn0_{i,j}, PIn1_{i,j} \leftarrow setPPBitGen(t0, t1)$\;

		$GIn0_{i,j} \leftarrow TT0_{i,j} \oplus PIn0_{i,j} $\;
		$GIn1_{i,j} \leftarrow TT1_{i,j} \oplus PIn1_{i,j} $\;

		GenSendToCloud( Permute([$TT0_{i,j},$ $TT1_{i,j}$]), permute\_bit\_locations )\;
	}
	GenSendToCloud($R_i$)\;
}

// Check: The cloud checks the gates to make sure the generator didn't cheat\;

\For{$i\leftarrow 0$ \KwTo $S$}{
	\If{isCheckCircuit(i)}{
		\For{$j\leftarrow 0$ \KwTo $len(savedWires)$}{
			// the cloud has received the truth table information, $TT0_{i,j}, TT1_{i,j}$, bit locations from $setPPBitGen$, and $R_i$\;

			$correct \leftarrow ( generateGateFromInfo() == receivedGateFromGen() )$\;

			// If any gate does not match, the cloud knows the generator tried to cheat.\;
			\If{!correct}{
				abort();\;
			}
		}
	}
}

// Evaluation\;

\For{$i\leftarrow 0$ \KwTo $S$}{
	\If{!isCheckCircuit(i)}{
		\For{$j\leftarrow 0$ \KwTo $len(savedWires)$}{

			//The cloud, using the previously saved $POutx_{i,j}$ value, and the location (point and permute) bit sent by the generator, creates $PInx_{i,j}$\;
			$PInx_{i,j} \leftarrow setPPBitEval( hash(R_i \oplus POutx_{i,j}) , location)$\;

			// Using $PInx_{i,j}$, the cloud selects the proper truth table entry $TTx_{i,j}$ from either $TT0_{i,j}$ or $TT1_{i,j}$ to decrypt\;

			// Creates $GInx_{i,j}$ to enter into the garbled circuit\;
			$GInx_{i,j} \leftarrow TTx_{i,j} \oplus POutx_{i,j}$\;
		}
	}
}

return GIn;
\end{algorithm}

\smallskip\noindent
{\it Generation}:

For $0 \leq i < S$, for $0 \leq j < len(savedWires)$ the generator creates a
{\it partial input gate}, which transforms a wire's saved values, $POut0_{i,j}$
and $POut1_{i,j}$, into values that can be used in the current garbled circuit
execution, $GIn0_{i,j}$ and $GIn1_{i,j}$.

For each circuit $0 \leq i < S$, the
generator creates a pseudorandom transformation value $R_i = \{0,1\}^K$, to assist with the
transformation. 




For each set of $POut0_{i,j}$ and $POut1_{i,j}$, the generator XORs each value
with $R_i$. Both results are then hashed, and put through a function to
determine the new permutation bit, as hashing removes the old permutation bit.

\vspace{-4mm}

\equsize \begin{eqnarray*}
t0 = hash(POut0_{i,j} \oplus R_i) \\ t1 = hash(POut1_{i,j}\oplus R_i)\\
PIn0_{i,j}, PIn1_{i,j} =setPPBitGen(t0, t1)
\end{eqnarray*} \normalsize


\noindent
This function, {\it setPPBitGen}, pseudo-randomly finds a bit that is different
between the two values of the wire and notes that bit to be the permutation bit.
{\it setPPBitGen} is seeded from $CSeed_i$, allowing the cloud to regenerate these values for the check
circuits.

For each $PIn0_{i,j}, PIn1_{i,j}$ pair, a set of values, $GIn0_{i,j}$ and $GIn1_{i,j}$, are created under the master key of $C_i$ -- where $CKey_i$ is the difference between 0 and 1 wire labels for the circuit. In classic garbled gate style, two truth table values, $TT0_{i,j}$ and $TT1_{i,j}$, are created such that:

\vspace{-4mm}

\equsize \begin{eqnarray*}
\quad\quad TT0_{i,j} \oplus PIn0_{i,j}  = GIn0_{i,j} \\ TT1_{i,j} \oplus PIn1_{i,j}  = GIn1_{i,j}
\end{eqnarray*} \normalsize


\noindent
The truth table, $TT0_{i,j}$ and $TT1_{i,j}$, is permuted so that the permutation bits of $PIn0_{i,j}$ and $PIn1_{i,j} $ tell the cloud which entry to select.  Each {\it partial input gate}, consisting of the permuted  $TT0_{i,j},$ $TT1_{i,j}$ values, the bit location from $setPPBitGen$, and each $R_i$, is sent to the cloud.

\smallskip\noindent
{\it Check}:

For all the check circuits, ({\it i.e.}, $\forall i : 0 \leq i < S$ where $isCheckCircuit(i)$ is true), for $0 \leq j < len(savedWires)$, the cloud receives the truth table information, $TT0_{i,j}, TT1_{i,j}$, and bit location from $setPPBitGen$, and proceeds to regenerate the gates based on the check circuit information. The cloud uses $R_i$ (sent by the generator), $POut0_{i,j}$ and  $POut1_{i,j}$ (saved during the previous execution), and $CSeed_i$ (recovered during the cut-and-choose) to generate the {\it partial input gates} in the same manner as described previously. The cloud then compares these gates to those the generator sent. If any gate does not match, the cloud knows the generator tried to cheat.

\smallskip\noindent
{\it Evaluation}:

For $0 \leq i < S$ where $!isCheckCircuit(i)$, for $0 \leq j < len(savedWires)$ the cloud receives the truth table information, $TTa_{i,j}, TTb_{i,j}$ and bit location from $setPPBitGen$. $a$ and $b$ are used to denote the two permuted truth table values. The cloud, using the previously saved $POutx_{i,j}$ value, creates the $PInx_{i,j}$ value

\vspace{-4mm}

\equsize \begin{eqnarray*}
\quad PInx_{i,j} = setPPBitEval( hash(R_i \oplus POutx_{i,j}) , location)
\end{eqnarray*} \normalsize


\noindent
 where $location$ is the location of the point and permute bit sent by the generator. Using the point and permute bit of $PInx_{i,j}$, the cloud selects the proper truth table entry $TTx_{i,j}$ from either $TTa_{i,j}$ or $TTb_{i,j}$ to decrypt, creates $GInx_{i,j}$ and then enters $GInx_{i,j}$ into the garbled circuit.

\vspace{-4mm}

\equsize \begin{eqnarray*}
GInx_{i,j} = TTx_{i,j} \oplus POutx_{i,j}
\end{eqnarray*} \normalsize


\smallskip\noindent
{\bf Phase 5: Circuit Generation and Evaluation}

\begin{algorithm}[h]

\scriptsize

\DontPrintSemicolon
 \SetKwData{Left}{left}\SetKwData{This}{this}\SetKwData{Up}{up}
  \SetKwFunction{algCutandChoose}{Cut\_and\_Choose}
\SetKwFunction{algEvaluatorInput}{Evaluator\_Input}
\SetKwFunction{algGeneratorInputCheck}{Generator\_Input\_Check}
\SetKwFunction{algPartialInput}{Partial\_Input}
\SetKwFunction{algCircuitExecution}{Circuit\_Execution}
\SetKwFunction{algCircuitOutput}{Circuit\_Output}
\SetKwFunction{algPartialOutput}{Partial\_Output}
  \SetKwInOut{Input}{Input}\SetKwInOut{Output}{Output}
\SetKwData{Left}{left}\SetKwData{This}{this}\SetKwData{Up}{up}
\Input{Generator\_Input, Evaluator\_Input, Partial\_Input}
\Output{Partial\_Output, Garbled\_Output}
\caption{Circuit\_Execution}
// The generator generates each garbled gate and sends it to the cloud. Depending on whether the circuit is a check or evaluation circuit, the cloud verifies that the gate is correct or evaluates the gate.\;

\For{$i\leftarrow 0$ \KwTo $S$}{
\For{$j\leftarrow 0$ \KwTo len(circuit)}{
		$g \leftarrow genGate(C_i,j)$\;
		$send(g)$\;
}
}

// the cloud receives all gates for all circuits, and then checks OR evaluates each circuit\;
\For{$i\leftarrow 0$ \KwTo $S$}{
\For{$j\leftarrow 0$ \KwTo len(circuit)}{
		$g \leftarrow recvGate()$\;
		\eIf{isCheckCircuit(i)}{
			\If{ ! $verifyCorrect(g)$ }{
				abort()\;
			}
		}{
			$eval(g)$\;
		}
}
}
return Partial\_Output, Garbled\_Output\;
\end{algorithm}


\noindent{\it Circuit Generation}:

The generator generates every garbled gate for each circuit and sends them to the cloud. Since the generator does not know the check and evaluation circuit split, nothing changes between the generation for check and evaluation circuits. For $0 \leq  i < S,$ For $0 \leq j <  len(circuit)$,

\vspace{-4mm}

\equsize
\begin{eqnarray*}
g = garbleGate(C_i,j)\\
send(g)
\end{eqnarray*}
\normalsize


\noindent{\it Circuit Evaluation and Check}:

The cloud receives garbled gates for all circuits. For evaluation circuits the cloud evaluates those garbled gates. For check circuits the cloud generates the correct gate, based on the circuit seed, and  is able to verify it is correct. For $0 \leq i < S,$ For $0 \leq j <  len(circuit)$, 


\vspace{-4mm}

\equsize
\begin{eqnarray*}
g = recvGate()\\
if(isCheckCircuit(i)) \quad  verifyCorrect(g)\\
else \quad eval(g)
\end{eqnarray*}
\normalsize

If a garbled gate is found not to be correct, the cloud informs the evaluator and generator of the incorrect gate and safely aborts. 

\smallskip\noindent
{\bf Phase 6: Output and Output Consistency Check}

\begin{algorithm}[h]

\scriptsize

\DontPrintSemicolon
 \SetKwData{Left}{left}\SetKwData{This}{this}\SetKwData{Up}{up}
  \SetKwFunction{algCutandChoose}{Cut\_and\_Choose}
\SetKwFunction{algEvaluatorInput}{Evaluator\_Input}
\SetKwFunction{algGeneratorInputCheck}{Generator\_Input\_Check}
\SetKwFunction{algPartialInput}{Partial\_Input}
\SetKwFunction{algCircuitExecution}{Circuit\_Execution}
\SetKwFunction{algCircuitOutput}{Circuit\_Output}
\SetKwFunction{algPartialOutput}{Partial\_Output}
  \SetKwInOut{Input}{Input}\SetKwInOut{Output}{Output}
\SetKwData{Left}{left}\SetKwData{This}{this}\SetKwData{Up}{up}
\Input{Garbled\_Output}
\caption{Circuit\_Output}

// a MAC of the output is generated inside the garbled circuit, and both the resulting garbled circuit output and the MAC are encrypted under a one-time pad.\;
$outEvlComplete = outEvl || MAC(outEvl)$\;
$result = ( outEvlMAC == MAC(outEvl) ) $\;

\If{!result}{
	abort() // output check fail\;
}
\end{algorithm}

As the final step of the garbled circuit execution, a MAC of the output is
generated inside the garbled circuit, based on a $k$-bit secret key entered into
the function.

\vspace{-4mm}

\equsize
\begin{eqnarray*}
outEvlComplete = outEvl || MAC(outEvl)
\end{eqnarray*}
\normalsize


\noindent
Both the resulting
garbled circuit output and the MAC are encrypted under the one-time pad (from phase 1 before) leaving the garbled circuit.

To receive output from the garbled circuit for any particular output bit $x$, a majority vote is taken across all evaluation circuits. For $0 \leq i < S$ where $!isCheckCircuit(i)$,

\vspace{-4mm}

\equsize
\begin{eqnarray*}
result = majority(COut_{0,x} \dots COut_{i-1,x} )
\end{eqnarray*}
\normalsize

Where $COut_{i,j}$ is the output bits, $i$ is the $i$th circuit and $j$ is the $j$th output bit from circuit $i$.

The cloud sends the
corresponding encrypted (under the one-time pad introduced in phase 1) output to each party.

The generator and evaluator then decrypt the received ciphertext by using
their one-time pad keys and
perform a MAC over real output to verify the cloud did not modify the
output by comparing the generated MAC with the MAC calculated within the garbled
circuit. 

\vspace{-4mm}

\equsize \begin{eqnarray*}
result = (outEvlMAC == MAC(outEvl))
\end{eqnarray*} \normalsize

Both parties, the generator and evaluator, now have their output.


\smallskip\noindent
{\bf Phase 7: Partial Output}

\begin{algorithm}[h]

\scriptsize

\DontPrintSemicolon
 \SetKwData{Left}{left}\SetKwData{This}{this}\SetKwData{Up}{up}
  \SetKwFunction{algCutandChoose}{Cut\_and\_Choose}
\SetKwFunction{algEvaluatorInput}{Evaluator\_Input}
\SetKwFunction{algGeneratorInputCheck}{Generator\_Input\_Check}
\SetKwFunction{algPartialInput}{Partial\_Input}
\SetKwFunction{algCircuitExecution}{Circuit\_Execution}
\SetKwFunction{algCircuitOutput}{Circuit\_Output}
\SetKwFunction{algPartialOutput}{Partial\_Output}
  \SetKwInOut{Input}{Input}\SetKwInOut{Output}{Output}
\SetKwData{Left}{left}\SetKwData{This}{this}\SetKwData{Up}{up}
\Input{Partial\_Output}
\caption{Partial\_Output}

\For{$i\leftarrow 0$ \KwTo S}{
\For{$j\leftarrow 0$ \KwTo len(Partial\_Output)}{
	//The generator saves both possible wire values \;
	GenSave($Partial\_Output0_{i,j}$)\;
	GenSave($Partial\_Output1_{i,j}$)\;
}
}

\For{$i\leftarrow 0$ \KwTo S}{
\For{$j\leftarrow 0$ \KwTo len(Partial\_Output)}{
	\eIf{ isCheckCircuit(i) } {
	EvlSave($Partial\_Output0_{i,j}$)\;
	EvlSave($Partial\_Output1_{i,j}$)\;
	}
	{ // circuit is evaluation circuit
		EvlSave($Partial\_OutputX_{i,j}$)\;
	}
}
}
\end{algorithm}

The generator saves both possible wire values for each partial output wire. For each evaluation circuit the cloud saves the partial output wire value. For check circuits the cloud saves both possible output values.

\subsection{Implementation}
\label{sec:design}

As with most garbled circuit systems there are two stages to our
implementation. The first stage is a compiler for creating garbled
circuits, while the second stage is an execution system to evaluate the
circuits.

We modified the compiler from Kreuter et al.~\cite{Kreuter2012} (hereon KSS12 compiler)   to allow for the saving of intermediate wire labels and loading wire  labels from a different SFE computation. By using the KSS12 compiler, 
we have an added benefit of being able to compare
circuits of almost identical size and functionality between our system and
CMTB, whereas other protocols compare circuits of sometimes vastly different sizes.

For our execution system, we started with the CMTB system and modified it according to our protocol requirements.
PartialGC automatically performs the output consistency check, and we
implemented this check at the circuit level.
We became aware and corrected issues with CMTB relating to too many primitive OT operations ($S*inputs$ instead $inputs$) performed in the outsourced oblivious transfer when using a high circuit parameter and too low a general security parameter ($log_2(input)$ instead of 80). The fixes reduced the run-time of the OOT, though the exact amount varied.

\section{Security of PartialGC}
\label{sec:sec}

In this section, we provide a proof of the PartialGC protocol, showing that our protocol preserves the standard security guarantees provided by traditional garbled circuits - that is, none of the parties learns anything about the private inputs of the other parties that is not logically implied by the output it receives. Section~\ref{sec:sec-sketch} provides a high-level overview of the proof. Section~\ref{sec:sec-model} goes over models and definitions, followed by security guarantees in Section~\ref{sec:sec-g} and a full proof in Section~\ref{sec:sec-full}.

\subsection{Proof Sketch}\label{sec:sec-sketch}

We know that the protocol described in CMTB allows us to garble individual circuits and securely outsource their evaluation. In this paper, we modify certain portions of the protocol to allow us to transform the output wire values from a previous circuit execution into input wire values in a new circuit execution. These transformed values, which can be checked by the evaluator, are created by the generator using circuit ``seeds.''

We also use some aspects of sS13, notably their novel cut-and-choose technique which ensures that the generator does not learn which circuits are used for evaluation and which are used for checking - this means that the generator must create the correct transformation values for all of the cut-and-choose circuits.

Because we assume that the CMTB garbled circuit scheme can securely garble any circuit, we can use it individually on the circuit used in the first execution and on the circuits used in subsequent executions. We focus on the changes made at the end of the first execution and the beginning of subsequent executions which are introduced by PartialGC.

The only difference between the initial garbled circuit execution and any other garbled circuit in CMTB is that the output wires in an initial PartialGC circuit are stored by the cloud, and are not delivered to the generator or the evaluator. This prevents them from learning the output wire labels of the initial circuit, but cannot be less secure than CMTB, since no additional steps are taken here.

Subsequent circuits we wish to garble differ from ordinary CMTB garbled circuits only by the addition, before the first row of gates, of a set of partial input gates. These gates don't change the output along a wire, but differ from normal garbled gates in that the two possible labels for each input wire are not chosen randomly by the generator, but are derived by using the two labels along each output wire of the initial garbled circuit.

This does not reduce security. In PartialGC, the input labels for partial input gates have the same property as the labels for ordinary garbled input gates: the generator knows both labels, but does not know which one corresponds to the evaluator's input, and the evaluator knows only the label corresponding to its input, but not the other label. This is because the evaluator's input is exactly the output of the initial garbled circuit, the output labels of which were saved by the evaluator. The evaluator does not learn the other output label for any of the output gates because the output of each garbled gate is encrypted. If the evaluator could learn any output labels other than those which result from an evaluation of the garbled circuit, the original garbled circuit scheme itself would not be secure.

The generator, which also generated the initial garbled circuit, knows both possible input labels for all partial evaluation gates, because it has saved both potential output labels of the initial circuit's output gates. Because of the outsourced oblivious transfer used in CMTB, the generator did not know which input labels to use for the initial garbled circuit, and therefore will not have been able to determine the output labels for that circuit. Therefore, the generator will likewise not know which input labels are being used for subsequent garbled circuits.

\subsection{Model and definitions}\label{sec:sec-model}

Throughout our protocol, we assume that none of the parties involved ever collude with the cloud. It is known that theoretical limitations exist when considering collusion in secure multiparty computation, and other schemes considering secure computation with multiple parties require similar restrictions on who and how many parties may collude while preserving security. If an outsourcing protocol is secure when both the generator and the cloud are malicious and colluding, this implies a secure two-party scheme where one party has sub-linear work with respect to the size of the circuit, which is currently only possible with fully homomorphic encryption~\cite{Kamara2012}. However, making the assumption that the cloud will not collude with the participating parties makes outsourcing securely a theoretical possibility.

While it is unlikely that a reputable cloud provider would allow external parties to illegally control or modify computations within their systems, we cannot assume the cloud will automatically be semi-honest. For example, our protocol requires a number of consistency checks to be performed by the cloud that ensure the participants do not cheat. Without mechanisms to force the cloud to make these checks, a ``lazy'' cloud provider could save resources by simply returning that all checks verified without actually performing them.

The work of Kamara et al.~\cite{Kamara2012} formalizes the idea of a non-colluding cloud based on the ideal-model/real-model security definitions common in secure multiparty computation. We apply their definitions to our protocol (for the two-party case in particular) as described below.

\smallskip
\noindent
{\bf Real-model execution.} The protocol takes place between two parties $(P_1,P_2)$ executing the protocol and a server $P_3$, where each of the executing parties provides input $x_i$, auxiliary input $z_i$, and random coins $r_i$. The server provides only auxiliary input $z_3$ and random coins $r_3$. There exists some subset of independent parties $(A_1,..A_m),m \leq 3$ that are malicious adversaries. Each adversary corrupts one executing party and does not share information with other adversaries. For all honest parties, let $OUT_i$ be its output, and for corrupted parties let $OUT_i$ be its view of the protocol execution. The $i^{th}$ partial output of a real execution is defined as $REAL^{(i)}(k,x;r)= \lbrace OUT_j : j \in H \rbrace \cup OUT_i$, where $H$ is the set of honest parties and $r$ is all random coins of all players.

\smallskip
\noindent
{\bf Ideal-model execution.} In the ideal model, the setup of participants is the same except that all parties are interacting with a trusted party that evaluates the function. All parties provide inputs $x_i$, auxiliary input $z_i$, and random coins $r_i$. If a party is semi-honest, it provides its actual inputs to the trusted party, while if the party is malicious or non-colluding, it provides arbitrary input values. In the case of the server $P_3$, this means simply providing its auxiliary input and random coins, as no input is provided to the function being evaluated. Once the function is evaluated by the trusted third party, it returns the result to the parties $P_1$ and $P_2$, while the server $P_3$ does not receive the output. If a party aborts early or sends no input, the trusted party immediately aborts. For all honest parties, let $OUT_i$ be its output to the trusted party, and for corrupted parties let $OUT_i$ be some value output by $P_i$. The $i^{th}$ partial output of an ideal execution in the presence of some set of independent simulators is defined as $IDEAL^{(i)}(k,x;r)= \lbrace OUT_j : j \in H \rbrace \cup OUT_i$ where H is the set of honest parties and r is all random coins of all players.

\begin{define}
A protocol securely computes a function f if there exists a set of probabilistic polynomial-time (PPT) simulators $\lbrace Sim_i \rbrace_{i \in [3]}$ such that for all PPT adversaries $(A_1, \dots , A_3), x, z$, and for all $i \in [3]$, we have 

$ \lbrace REAL^{(i)} (k, x; r)\rbrace_{k \in N} \approx \lbrace IDEAL^{(i)}(k, x; r) \rbrace_{k \in N}$

\noindent
Where $S = (S_1, \dots , S_3), S_i = Sim_i(A_i)$, and $r$ is random and uniform.
\end{define}

\subsection{Security Guarantees}\label{sec:sec-g}

\noindent
{\bf Generator's Input Consistency Check}

During the cut-and-choose, multiple copies of the garbled circuit are constructed and then either checked or evaluated. A malicious generator may provide inconsistent inputs to different evaluation circuits. For some functions, it is possible to use inconsistent inputs to extract information of Eval's input~\cite{Lindell2007}.

\begin{claim}
The generator in our protocol cannot trick the evaluator into using different inputs for different evaluation circuits with greater than negligible probability.
\end{claim}

We use the generator's input consistency check from \cite{shelat13}, and defer to the proof provided in that paper, noting that simulators $S_1$ and $S_2$ can be constructed such that any malicious generator (resp. evaluator) cannot tell whether it is working with $S_1$ (resp. $S_2$) in the ideal model, or with an honest evaluator (resp. generator) in the real model.

We further note there is no problem with allowing the cloud to perform this check; for the generator's inconsistent input to pass the check, the cloud would have to see the malicious input and ignore it, which would violate the non-collusion assumption.

\noindent
{\bf Validity of Evaluator Inputs}

To ensure that the generator cannot learn anything about the evaluator's inputs by corrupting the garbled values sent during the OT, we use from CMTB the random input encoding technique by Lindell and Pinkas~\cite{Lindell2007}. This technique allows the evaluator to encode each input bit as the XOR of a set of input bits. Thus, if the generator corrupts one of those input bits as in a selective failure attack, it reveals nothing about the evaluator's true input. Additionally, we use the commitment technique employed by Kreuter et al.~\cite{Kreuter2012} to ensure that the generator cannot swap garbled input wire labels between the zero and one value. To accomplish this, the generator commits to the wire labels before the cut and choose. During the cut and choose, the input labels for the check circuits are opened to ensure that they correspond to only one value across all circuits. Then, during the OOT, the commitment keys for the labels that will be evaluated are sent instead of the wire labels themselves. Because our protocol implements this technique directly from previous work, we do not make any additional claims of security.

\noindent
{\bf Correctness of Saved Values}

Scenarios where either party enters incorrect values in the next computation reduce to previously solved problems in garbled circuits. If the generator does not use the correct values, then it reduces to the problem of creating an incorrect garbled circuit. If the evaluator does not use the correct saved values then it reduces to the problem of the evaluator entering garbage values into the garbled circuit execution; this would be caught by the output consistency check. 

\noindent
{\bf Garbled Circuit Generation}

To ensure the evaluated circuits are generated honestly, we require two properties. First, we limit the generator's ability to trick the evaluator into evaluating a corrupted circuit using a cut-and-choose technique similar to a typical, two-party garbled circuit evaluation. Second, we ensure that a lazy Cloud attempting to conserve system resources cannot bypass the circuit checking step without being discovered.

\begin{claim}
Security: Assuming that the hash function $UHF( x )$ (as used in phase 3) is a one-way, collision-resistant hash and that the commitment scheme used is fully binding, then the generator has at best a $2^{-k}$ probability of tricking the evaluator into evaluating a majority of corrupted circuits, where $k$ is the number of circuits generated.
\end{claim}

This claim follows directly from sS13. The probability of the generator finding a hash collision and thus fooling the evaluator is at most $1/\left\vert B\right\vert$, where $B$ is the range of the hash function.

\begin{claim}
Proof-of-work: Assuming the hash function is one-way and collision resistant, the Cloud has a negligible probability of producing a check hash that passes the seed check without actually generating the check circuit.
\end{claim}

As previously stated, before the circuit check begins the generator sends the evaluator $k$ hashed circuit values $H_1(GC_i)$. Once the evaluation circuits are selected, the cloud must generate some circuits and hash them into check hashes $H_1(GC'_i)$. If the cloud attempts to skip the generation of the check circuits, it must generate hash values $H'_i = H_i$ for $i \in Chk$. Based on security guarantees of the hash, and the non-collusion property, the cloud has a negligible probability of correctly generating these hash values.

\noindent
{\bf Abort on Check Failure}

If any of the check circuits fail, the cloud reports the incorrect check circuit to both the generator and evaluator. At this point, the remaining computation and any saved values must be abandoned. However, as is standard in SFE, the cloud cannot abort on an incorrect evaluation circuit even when it is known to be incorrect.

\noindent
{\bf Concatenation of Incorrect Circuits}

If the generator produces a single incorrect circuit and the cloud does not abort, the generator learns that the circuit was used for evaluation, and not as a check circuit. This leaks no information about the input or output of the computation; to do that, the generator must corrupt a majority of the evaluation circuits without modifying a check circuit. An incorrect circuit that goes undetected in one execution has no effect on subsequent executions as long the total amount of incorrect circuits is less than the majority of evaluation circuits.

\noindent
{\bf Using Multiple Evaluators}

One of the benefits of our outsourcing scheme is that the state is saved at the generator and cloud allowing the use of different evaluators in each computation. Previously, it was shown a group of users working with a single server using 2P-SFE was not secure against malicious adversaries, as a malicious server and last $k$ parties, also malicious, could replay their portion of the computation with different inputs and gain more information than they can with a single computation~\cite{Halevi11}. However, this is not a problem in our system as at least one of our servers, either the generator or cloud, must be semi-honest due to non-collusion, which obviates the attack stated above.

\noindent
{\bf Threat Model}

As we have many computations involving the same generator and cloud, we have to extend the threat model for how the parties can act in different computations. There can be no collusion in each singular computation. However, the malicious party can change between computations as long as there is no chain of malicious users that link the generator and cloud -- this would break the non-collusion assumption.

\subsection{Proof}\label{sec:sec-full}

We formally prove the security of our protocol with the following theorem, which gives security guarantees identical to that of CMTB and the protocol by Kamara et al.~\cite{Kamara2012}.

\begin{theorem}
The outsourced two-party SFE protocol securely computes a function $f(a,b)$ in the following two corruption scenarios: (1)The Cloud is malicious and non-cooperative with respect to the rest of the parties, while all other parties are semi-honest, and (2)All but one party is malicious, while the Cloud is semi-honest.
\end{theorem}

\begin{proof}
To demonstrate that $\lbrace REAL^{(i)}(k, x; r)\rbrace_{k \in N} \approx \\ \lbrace IDEAL^{(i)}(k, x; r) \rbrace_{k \in N}, \forall i \in \lbrace generator, evaluator, cloud \rbrace$, we consider separately each case where a party deviates from the protocol, and then show that these cases reveal no additional information, by considering each point at which the parties interact.
\end{proof}

For the remaining portion of the proof, we shall call the generator, the evaluator, and the cloud, as A, B, and C, respectively. To denote that a party is malicious, we shall use A*, B*, and C*, respectively.

\subsubsection{Malicious Evaluator}

Both the generator and the cloud participate honestly in the protocol. During the protocol execution, the evaluator only exchanges messages with the other participants at certain points: the cut-and-choose, the oblivious transfer, sending decryption information at the end of the OOT, checking the generator's input consistency, and receiving the proof of validity and output from the garbled circuit. Thus, our simulator need only ensure that these sections of the protocol are indistinguishable to the adversary - e.g., we need not consider phase 4. Our proof comprises the following hybrid experiments and lemmas:

\smallskip
\noindent
{\bf Simulating the random number generation / coin-flip}

$Hybrid1^{(A)}(k, x; r)$: This experiment is the same as \\ $REAL^{(A)}(k, x; r)$ except that instead of running a fair coin toss protocol with A*, the experiment chooses a random string $\rho$, and a coin-flipping simulator $S_{coinFlip} (\rho, 1^k )$ produces the protocol messages that output $\rho$.

\begin{lemma}
$REAL^{(A)}(k, x; r) \approx Hybrid1^{(A)}(k, x; r)$
\end{lemma}

\begin{proof}
Based on the security of the fair coin toss protocol, we know that there exists a simulator $S_{coinFlip}( \cdot , \cdot )$ such that an interaction with $S_{coinFlip}( \cdot , \cdot )$ is indistinguishable from a real protocol interaction. Since everything else in $Hybrid1^{(A)}(k, x; r)$ is exactly the same as in $REAL^{(A)}(k, x; r)$, this proves the lemma.
\end{proof}

\noindent
{\bf Simulating the primitive OT}

$Hybrid2^{(A)}(k, x; r)$: This experiment is the same as \\$Hybrid1^{(A)}(k, x; r)$ except that during the Outsourced Oblivious Transfer, the experiment invokes a simulator $S_{OT}$ to simulate the primitive oblivious transfer operation with A*. The simulator sends A* a random string s and receives the columns of the matrix Q*.

\begin{lemma}
$Hybrid1^{(A)}(k, x; r) \approx Hybrid2^{(A)}(k, x; r)$
\end{lemma}

\begin{proof}
Based on the malicious security of the OT primitive, we know that there exists a simulator $S_{OT}$ such that an interaction with this simulator is indistinguishable from a real execution of the oblivious transfer protocol. Since everything else in $Hybrid2^{(A)}(k, x; r)$ is identical to $Hybrid1^{(A)}(k, x; r)$, this proves the lemma.
\end{proof}

\noindent
{\bf Checking the output of OOT}

$Hybrid3^{(A)}(k, x; r)$: This experiment is the same as \\$Hybrid2^{(A)}(k, x; r)$ except that the experiment aborts if the matrix Q* is not formed correctly (that is, if A* used inconsistent input values $ea^*$ for any column in generating Q*).

\begin{lemma}
$Hybrid2^{(A)}(k, x; r) \approx Hybrid3^{(A)}(k, x; r)$
\end{lemma}

\begin{proof}
Let us consider a case in $Hybrid2^{(A)}(k, x; r)$ where for some value of $i$, A* sends the column value $T^i  \oplus ea'$ for some $ea' \neq ea^*$ such that the $i^{th}$ bit is $b$ in $ea^*$ and $b \oplus 1$ in $ea'$. Then, for every row in Q*, the $i^{th}$ bit will be encrypted in the $b \oplus 1$ entry.

However, when A* sends the value $ea^* \oplus p^*$ to the cloud for decryption, the cloud will decrypt the $i^{th}$ choice, and then it will decrypt the $b \oplus 1$ entry instead of the $b$ entry. This will result in an invalid decryption with probability $1 - \epsilon$ for a negligible value of $\epsilon$. This decryption is not (with high probability) a valid commitment key, so the garbled input values won't de-commit, causing the cloud to abort. In $Hybrid3^{(A)}(k, x; r)$, since the experiment observes the messages Q* , p* , and $ea^* \oplus p^*$ , it can recover $ea^*$ and check Q* for consistency. (The cloud always aborts if an inconsistency is detected.)
\end{proof}

\noindent
{\bf Simulating consistency check and substituting inputs}

$Hybrid4^{(A)}(k, x; r)$: This experiment is the same as \\$Hybrid3^{(A)}(k, x; r)$ except that the experiment provides a string of $2 \cdot n$ zeros, denoted $0^{2 \cdot n}$, during the consistency check to replace the generator's input.

\begin{lemma}
$Hybrid3^{(A)}(k, x; r) \approx Hybrid4^{(A)}(k, x; r)$
\end{lemma}

\begin{proof}
Here we cite the proof of shelat and Shen's scheme~\cite{shelat13}. Since the messages sent in our scheme are identical to theirs in content, we simply change the entity sending the message in the experiment and the lemma still holds.
\end{proof}

\noindent
{\bf Output and output consistency check}

$Hybrid5^{(A)}(k, x; r)$: This experiment is identical to \\$Hybrid4^{(A)}(k, x; r)$ except that instead of returning the output of the circuit, the experiment provides A* with the result sent from the trusted external oracle.

\begin{lemma}
$Hybrid4^{(A)}(k, x; r) \approx Hybrid5^{(A)}(k, x; r)$
\end{lemma}

\begin{proof}
Based on the security of garbled circuits, the trusted third party output and the circuit output will be indistinguishable when provided with the input of A*, $ea^*$. Since we have the k-bit secret key and know the output, we can produce the necessary MAC under the one time pad (similar to Hybrid1). Simulating a majority vote is trivial; the evaluator then decrypts the ciphertext using its OTP key, and then performs a MAC on the output, and compares with the MAC produced within the garbled circuit. This, again, is indistinguishable from the circuit version since we know the output and can ensure that the MACs are identical.
\end{proof}

\subsubsection{Malicious Generator}

Here, both the evaluator and the cloud participate honestly in the protocol.

\noindent
{\bf Simulating the cut and choose}
$Hybrid1^{(B)}(k, x; r)$: This experiment is the same as $REAL^{(B)}(k, x; r)$, except during the oblivious transfer (the generator offers up decryption keys for his input as well as for the circuit seed and possible evaluator's input for each circuit), where we randomly choose check and evaluation circuits.
\begin{lemma}
$REAL^{(B)}(k, x; r) \approx Hybrid1^{(B)}(k, x; r)$
\end{lemma}

\begin{proof}
Since the generator does not learn the circuit split because of the security of the oblivious transfer, this portion of the protocol is indistinguishable from the real version. We can perform all the checks and decryptions required; we also have all the keys needed to check whether we should abort, so a selective failure attack does not help distinguish between the experiment and the real version.
\end{proof}

\noindent
{\bf Simulating the OT}
$Hybrid2^{(B)}(k, x; r)$: This experiment is the same as $Hybrid1^{(B)}(k, x; r)$ except that rather than run the primitive oblivious transfer with A, the experiment generates a random input string $ea'$ and a random matrix T, then runs a simulator $S_{OT}$ with B*, which gives B* exactly one element from the pair $(T^i , T^i \oplus ea')$ depending on the  $i^{th}$ selection bit of B*.

\begin{lemma}
$Hybrid1^{(B)}(k, x; r) \approx Hybrid2^{(B)}(k, x; r)$
\end{lemma}

\begin{proof}
Based on the security of the primitive OT scheme, we know that the simulator $S_{OT}$ exists, that it can recover B*'s selection bits from the interaction, and that an interaction with it is indistinguishable from a real execution of the OT. Since B* cannot learn any distinguishing information from A's input, again based on the security of the OT primitive, indistinguishability holds between the two experiments.
\end{proof}

\noindent
{\bf Checking the output of OOT}
$Hybrid3^{(B)}(k, x; r)$: This experiment is the same as $Hybrid2^{(B)}(k, x; r)$ except that the experiment checks the validity of B*'s output from the OOT. Since the experiment possesses $T$, $ea'$, and $s^*$ (which was recovered by the oblivious transfer simulator $S_{OT}$ in the previous hybrid), the experiment can check whether or not the encrypted set of outputs Y* is well-formed. If not, the experiment aborts.

\begin{lemma}
$Hybrid2^{(B)}(k, x; r) \approx Hybrid3^{(B)}(k, x; r)$
\end{lemma}
\begin{proof}
Recall that in $Hybrid2^{(B)}(k,x;r)$, if B* does not format the output of the OOT correctly, the cloud will fail to recover a valid commitment key with probability $1 - \epsilon$, where $\epsilon$ is negligible in the security parameter. In this case, the committed garbled circuit labels will fail to decrypt properly and the cloud will abort the protocol. In $Hybrid3^{(B)}(k, x; r)$, since the experiment has seen the values $Q$, $s^*$, $ea'$, and Y*, it can trivially check to see whether Y* is correctly formed; we about on failure, so a selective failure attack does not help distinguish between the experiment and the real version. Further, B* cannot swap any of A's input labels in the commitments~\cite{shelat13}.
\end{proof}

\noindent
{\bf Checking input consistency and recovering inputs}

$Hybrid4^{(B)}(k,x;r)$: This experiment is the same as \\$Hybrid3^{(B)}(k,x;r)$ except that the experiment recovers B*'s input b* during the input consistency check using the random seed recovered in $Hybrid1^{(B)}(k,x;r)$. If the consistency check does not pass or if B*'s input cannot be recovered, the experiment immediately aborts.
\begin{lemma}
$Hybrid3^{(B)}(k, x; r) \approx Hybrid4^{(B)}(k, x; r)$
\end{lemma}

\begin{proof}
The experiment can perform a hash of the generator's input for each evaluation circuit (we know the split from Hybrid1). If any of these hashes are different, then we know that the generator tried to cheat~\cite{shelat13} and can abort if necessary - this is done by the cloud in the real version of the protocol.
\end{proof}

\noindent
{\bf Generating, checking, and evaluating partial input gates}

$Hybrid5^{(B)}(k,x;r)$: This experiment is the same as \\$Hybrid4^{(B)}(k,x;r)$, except for the following: for check circuits, the experiment uses the data sent by the generator, as well as $POut0_{i,j}$, $POut1_{i,j}$, and $CSeed_i$ (recovered during the cut-and-choose) to generate the partial input gates in the same manner as shown in phase 4. It then compares these gates to those the generator sent. If any gate does not match, we know the generator tried to cheat. For evaluation circuits, it finds the point and permute bit and produces the value $GInx_{i,j}$ (this is needed for future hybrids).
\begin{lemma}
$Hybrid4^{(B)}(k, x; r) \approx Hybrid5^{(B)}(k, x; r)$
\end{lemma}

\begin{proof}
The experiment already has all the $POut$ values from previous hybrids, and receives \\$R_i, TT0_{i,j},TT1_{i,j}$,  and bit location ($setPPBitGen$) from B*. It also has $CSeed_i$ from Hybrid1 - thus, it can generate the partial input gates as described in phase 4, and perform all the necessary checks. Since we also know the split from Hybrid1, we can also generate $GInx_{i,j}$ values for the evaluation circuits (which are then entered into the garbled circuit by the cloud in the real version of the protocol).
\end{proof}

\noindent
{\bf Simulating the output check}
$Hybrid6^{(B)}(k, x; r)$: This experiment is the same as $Hybrid5^{(B)}(k,x;r)$ except that during the output phase the experiment prepares the result received from the trusted third party as the output instead of the output from the circuit.

\begin{lemma}
$Hybrid5^{(B)}(k, x; r) \approx Hybrid6^{(B)}(k, x; r)$
\end{lemma}

\begin{proof}
Based on the security of garbled circuits, the trusted third party output and the circuit output will be indistinguishable when provided with the input of B*, $b^*$. Since we have the k-bit secret key and know the output, we can produce the necessary MAC under the one time pad (similar to Hybrid1). Simulating a majority vote is trivial; the experiment then decrypts the ciphertext using its OTP key, and then performs a MAC on the output, and compares with the MAC produced within the garbled circuit. This, again, is indistinguishable from the circuit version to B*, since we know the appropriate OTP key.
\end{proof}

\subsubsection{Malicious Cloud}

Here, both the generator and the evaluator participate honestly in the protocol.

\noindent
{\bf Replacing inputs for the OOT}
$Hybrid1^{(C)}(k, x; r)$: This experiment is the same as $REAL^{(C)}(k, x; r)$ except that during the OOT, the experiment replaces A's input $ea$ with a string of zeros $ea' = 0^{2 \cdot l \cdot n}$. This value is then used to select garbled input values from B in the OOT, which are then forwarded to C* according to the protocol.
\begin{lemma}
$REAL^{(C)}(k, x; r) \approx Hybrid1^{(C)}(k, x; r)$
\end{lemma}
\begin{proof}
In a real execution, C* will observe the random matrix T, the encrypted commitment keys Y, and A's
input XOR'd with the permutation string $ea \oplus p$. Since $p$ is random, $ea \oplus p$ is indistinguishable from $p$. Since T is randomly generated in both $REAL^{(C)}(k, x; r)$ and $Hybrid1^{(C)}(k, x; r)$, they are trivially indistinguishable. Considering the output pairs Y, half of the commitment keys (those not selected by A) will consist of values computationally indistinguishable from random (since they are XOR'd with a hash), and the keys can only be recovered if C* can find a collision with the hash value $H(j, s)$ without having B's random value s. Thus, C* cannot distinguish an execution of OOT with A's input $ea$ and the simulator's input replacement $ea'$. Since the rest of the protocol follows $REAL^{(C)}(k, x; r)$ exactly, this proves the lemma.
\end{proof}

\noindent
{\bf Replacing inputs for the consistency check}\\
$Hybrid2^{(C)}(k, x; r)$: This experiment is the same as \\$Hybrid1^{(C)}(k, x; r)$ except that the experiment replaces B's input b with all zeros $0^{2 \cdot n}$. This value is then prepared and checked according to the protocol for consistency across evaluation circuits.

\begin{lemma}
$Hybrid1^{(C)}(k, x; r) \approx Hybrid2^{(C)}(k, x; r)$
\end{lemma}
\begin{proof}
In this hybrid, C* observes a set of garbled input wire values from B. Based on the security of garbled circuits, observing one set of garbled input wire values is indistinguishable from observing any other set of input wire values, such that C* cannot distinguish between the garbled input for b and the garbled input for $0^{2 \cdot n}$. The rest of the hybrid is the same as $Hybrid1^{(C)}(k, x; r)$.
\end{proof}

\noindent
{\bf Partial gates}
$Hybrid3^{(C)}(k, x; r)$: This experiment is the same as $Hybrid2^{(C)}(k, x; r)$ except that the experiment sends the appropriate truth table information, $R_i$, $POut$ values, and the bit location from $setPPBitGen$ to C*.

\begin{lemma}
$Hybrid2^{(C)}(k, x; r) \approx Hybrid3^{(C)}(k, x; r)$
\end{lemma}
\begin{proof}
In this hybrid, C* observes a set of garbled wire and truth table values from B. Since these partial gates are identical to garbled circuit gates as far as operation is concerned, the security of garbled circuits ensures that observing one set of garbled input wire values is indistinguishable from observing any other set of wire values, such that C* cannot distinguish between the garbled value for b and the garbled value for $0^{2 \cdot n}$. The rest of the hybrid is the same as $Hybrid2^{(C)}(k, x; r)$.
\end{proof}

\noindent
{\bf Checking the output}

$Hybrid4^{(C)}(k, x; r)$: This experiment is the same as \\$Hybrid3^{(C)}(k, x; r)$ except that after the circuit is evaluated, the experiment checks that the result output by C* is as expected; otherwise, the experiment immediately aborts.

\begin{lemma}
$Hybrid3^{(C)}(k, x; r) \approx Hybrid4^{(C)}(k, x; r)$
\end{lemma}

\begin{proof}
After the cloud sends the output (under the OTP from phase 1) to us, we can decrypt, since we have the OTP from previous hybrids. We then perform a MAC and compare with the MAC calculated within the garbled circuit to verify that C* did not modify the output.
\end{proof}

\section{Performance Evaluation}
\label{sec:experiments}

We now demonstrate the efficacy of PartialGC through a comparison with the CMTB
outsourcing system. Apart from the performance gains from using cut-and-choose
from sS13, PartialGC provides other benefits through generating partial input
values after the first execution of a program. On subsequent executions, the
partial inputs act to amortize overall costs of execution and bandwidth.

We demonstrate that the evaluator in the system can be a mobile device
outsourcing computation to a more powerful system. We also show that other
devices, such as server-class machines, can act as evaluators, to show the
generality of this system. Our testing environment includes a 64-core server
containing 1 TB of RAM, which we use to model both the Generator and Outsourcing
Proxy parties. We run separate programs for the Generator and Outsourcing Proxy,
giving them each 32 threads. For the evaluator, we use a Samsung Galaxy Nexus
phone with a 1.2 GHz dual-core ARM Cortex-A9 and 1 GB of RAM running Android
4.0, connected to the server through an 802.11 54 Mbps WiFi in an isolated
environment. In our testing we also use a single server process as the evaluator. For these tests we create that process on our 64-core server as well. We ran the CMTB
implementation for comparison tests under the same setup.

\subsection{Execution Time}

The PartialGC system is particularly well suited to complex computations that
require multiple stages and the saving of intermediate state. Previous garbled
circuit execution systems have focused on single-transaction evaluations, such
as computing the ``millionaires'' problem (i.e., a joint evaluation of which
party inputs a greater value without revealing the values of the inputs) or
evaluating an AES circuit.

Our evaluation considers two comparisons: the improvement of our system
compared with CMTB without reusing saved values, and comparing our protocol
for saving and reusing values against CMTB if such reuse was implemented in
that protocol. We also benchmark the overhead for saving and loading values
on a per-bit basis for 256 circuits, a necessary number to achieve a
security parameter of $2^{-80}$ in the malicious model. In all cases, we run
10 iterations of each test and give timing results with 95\% confidence
intervals. Other than varying the number of circuits our system parameters are set for 80-bit security.

The programs used for our evaluation are exemplars of differing input sizes
and differing circuit complexities:

\noindent {\bf Keyed Database:} In this program, one party enters a database
and keys to it while the other party enters a key that indexes into the
database, receiving a database entry for that key. This is an example of a
program expressed as a small circuit that has a very large amount of
input.

\noindent {\bf Matrix Multiplication:} Here, both parties enter 32-bit
numbers to fill a matrix.  Matrix multiplication is performed before the
resulting matrix is output to both parties. This is an example of a program with a large amount of inputs with a large circuit.

\noindent {\bf Edit (Levenstein) Distance:} This program finds the distance
between two strings of the same length and returns the difference. This is an
example of a program with a small number of inputs and a medium sized circuit.

\noindent {\bf Millionaires:} In this classic SFE program, both parties
enter a value, and the result is a one-bit output to each party to let them
know whether their value is greater or smaller than that of the other party.
This is an example of a small circuit with a large amount of input.

\begin{table*}[t]
\centering
\small
\begin{tabular}{ |c | c | c| c| c| c| c| c| c|c| }
\hline
 & CMTB & PartialGC  \\ \hline

KeyedDB 64 & 6,080 & 20,891  \\\hline
KeyedDB 128 & 12,160 & 26,971 \\\hline
KeyedDB 256 & 24,320 & 39,131  \\\hline
MatrixMult8x8 & 3,060,802 & 3,305,113 \\\hline
Edit Distance 128 & 1,434,888  & 1,464,490  \\\hline
Millionaires 8192 &  49,153 & 78,775  \\\hline
LCS Incremental 128  & 4,053,870  & 87,236  \\\hline
LCS Incremental 256  & 8,077,676  & 160,322  \\\hline
LCS Incremental 512  & 16,125,291  & 306,368  \\\hline
LCS Full 128  & 2,978,854  & -  \\\hline
LCS Full 256  & 13,177,739  & -  \\\hline
\multicolumn{3}{c}{}
\end{tabular}
\caption{Non-XOR gate counts for the various circuits. In the first 6
circuits, the difference between CMTB and PartialGC gate counts is in the
consistency checks. The explanation for the difference in size between the incremental versions of longest common substring (LCS) is given in  {\it Reusing Values}.}
\label{table:gatecounts}
\end{table*}

Gate counts for each of our programs can be found in
Table~\ref{table:gatecounts}. The only difference for the programs described
above is the additional of a MAC function in PartialGC. We discuss the reason for
this check in Section~\ref{sub:dis}.


\begin{table*}[t]
\centering
\small
\begin{tabular}{ |c | c | c| c| c| c| c| c| c|c| }
\hline
  & \multicolumn{3}{|c|}{ 16 Circuits}  & \multicolumn{3}{|c|}{ 64 Circuits}  & \multicolumn{3}{|c|}{ 256 Circuits}  \\
\hline
 & CMTB & PartialGC &  & CMTB & PartialGC&  & CMTB & PartialGC &  \\ \hline

KeyedDB 64 & 18 $\pm$ 2\% & 3.5 $\pm$ 3\% & 5.1x &72 $\pm$ 2\% & 8.3 $\pm$ 5\% & 8.7x &290 $\pm$ 2\% & 26 $\pm$ 2\% & 11x \\\hline
KeyedDB 128 & 33 $\pm$ 2\% & 4.4 $\pm$ 8\% & 7.5x &140 $\pm$ 2\% & 9.5 $\pm$ 4\% & 15x &580 $\pm$ 2\% & 31 $\pm$ 3\% & 19x \\\hline
KeyedDB 256 & 65 $\pm$ 2\% & 4.6 $\pm$ 2\% & 14x &270 $\pm$ 1\% & 12 $\pm$ 6\% & 23x &1200 $\pm$ 3\% & 38 $\pm$ 5\% & 32x \\\hline
MatrixMult8x8 & 48 $\pm$ 4\% & 46 $\pm$ 4\% & 1.0x &110 $\pm$ 8\% & 100 $\pm$ 7\% & 1.1x &400 $\pm$ 10\% & 370 $\pm$ 5\% & 1.1x \\\hline
Edit Distance 128  & 21 $\pm$ 6\% & 22 $\pm$ 3\% & 0.95x &47 $\pm$ 7\% & 50 $\pm$ 9\% & 0.94x &120 $\pm$ 9\% & 180 $\pm$ 6\% & 0.67x \\\hline
Millionaires 8192 & 35 $\pm$ 3\% & 7.3 $\pm$ 6\% & 4.8x &140 $\pm$ 2\% & 20 $\pm$ 2\% & 7.0x &580 $\pm$ 1\% & 70 $\pm$ 2\% & 8.3x \\\hline
\multicolumn{10}{c}{}

\end{tabular}
\caption{Timing results comparing PartialGC to CMTB without saving any values. All times in seconds.}
\label{table:keyeddbphoneserver}
\end{table*}

Table~\ref{table:keyeddbphoneserver} shows the results from our experimental
tests. In the best case, execution time was reduced by a factor of 32 over
CMTB, from 1200 seconds to 38 seconds, a 96\% speedup over CMTB. Ultimately,
our results show that our system outperforms CMTB when the input checks are
the bottleneck. This run-time improvement is due to improvements we added
from sS13 and occurs in the keyed database, millionaires, and matrix
multiplications programs. In the other program, edit distance, the input
checks are not the bottleneck and PartialGC does not outperform CMTB. The
total run-time increase for the edit distance problem is due to overhead of
using the new sS13 OT cut-and-choose technique which requires sending each
gate to the evaluator for check circuits and evaluation circuits. This is
discussed further in Section~\ref{sub:dis}.  The typical use case we imagine
for our system, however, is more like the keyed database program, which has
a large amount of inputs and a very small circuit. We expand upon this use
case later in this section.

\noindent{\bf Reusing Values}

For a test of our system's wire saving capabilities we tested a dynamic programming problem, longest common substring, in both PartialGC and CMTB. This program determines the length of the longest common substring between two strings. Rather than use a single
computation for the solution, our version incrementally adds a single bit of input to both strings each time the computation is run and outputs the results each time to the evaluator. We believe this is a realistic comparison to a real-world application that incrementally adds data during each computation where it is faster to save the intermediate state and add to it after seeing an intermediate result than rerun the entire computation many times after seeing the result. 

For our testing, PartialGC uses our technique to reuse wire values. In CMTB, we
save each desired internal bit under a one-time pad and re-enter them into the
next computation, as well as the information needed to decrypt the ciphertext.
We use a MAC (the AES circuit of KSS12) to verify that the party saving the output
bits did not modify them. We also use AES to generate a one-time pad inside the
garbled circuit. We use AES as this is the only cryptographically secure
function used in CMTB. Both parties enter private keys to the MAC functions.
This program is labeled {\it CMTB-Inc}, for CMTB {\it incremental}. The size of
this program represents the size of the total strings. We also created a circuit
that computes the complete longest common substring in one computation labeled
{\it CMTB-Full}.

\begin{figure}[t]
\centering
\includegraphics[width=3in]{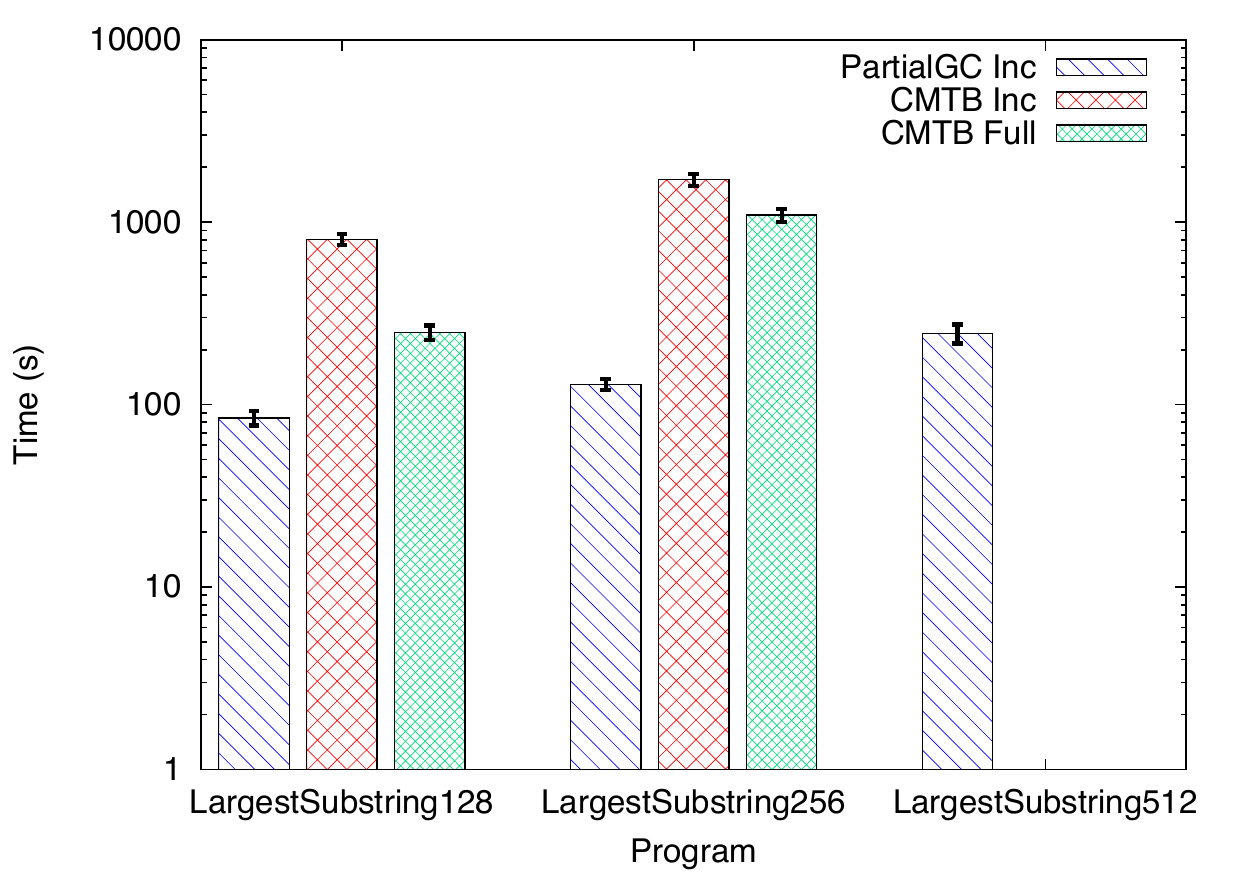}
\caption{Results from testing our largest common substring (LCS) programs for PartialGC and CMTB. This shows when changing a single input value is more efficient under PartialGC than either CMTB program. CMTB crashed on running LCS Incremental of size 512 due to memory requirements.  We were unable to complete the compilation of CMTB Full of size 512.}
\label{fig:largesub}
\end{figure}

The resulting size of the PartialGC and CMTB circuits are shown in
Table~\ref{table:gatecounts}, and the results are shown in
Figure~\ref{fig:largesub}. This result shows that saving and reusing values in
PartialGC is more efficient than completely rerunning the computation. The input
consistency check adds considerably to the memory use on the phone for {\em
CMTB-Inc} and in the case of input bit 512, the {\em CMTB-Inc} program will not
complete. In the case of the 512-bit {\em CMTB-Full}, the program would not
complete compilation in over 42 hours. In our {\em CMTB-Inc} program, we assume
the cloud saves the output bits so that multiple phones can have a shared
private key. 

Note that the growth of {\em CMTB-Inc} and {\em CMTB-Full} are
different. {\em CMTB-Full} grows at a larger rate (4x for each 2x factor
increase) than {\em CMTB-Inc} (2x for each 2x factor increase), implying
that although at first it seems more efficient to rerun the program if small
changes are desired in the input, eventually this will not be the case. Even with a more efficient AES function, {\em CMTB-Inc} would not be faster as the bottleneck is the input, not the size of the circuit.






The overhead of saving and reusing values is discussed further in Appendix~\ref{appendix:reuse}.

\begin{table*}[t]
\centering
\small
\begin{tabular}{ |c | r | r| r| c| c| c| c| c| c| c| c| c|}
\hline
  & \multicolumn{3}{|c|}{ 16 Circuits}  & \multicolumn{3}{|c|}{ 64 Circuits}  & \multicolumn{3}{|c|}{ 256 Circuits}  \\
\hline
 & CMTB & PartialGC &  & CMTB & PartialGC&  & CMTB & PartialGC &  \\ \hline

KeyedDB 64 & 6.6 $\pm$ 4\% & 1.4 $\pm$ 1\% & 4.7x &27 $\pm$ 4\% & 5.1 $\pm$ 2\% & 5.3x &110 $\pm$ 2\% & 24.9 $\pm$ 0.3\% & 4.4x \\\hline
KeyedDB 128 & 13 $\pm$ 3\% & 1.8 $\pm$ 2\% & 7.2x &54 $\pm$ 4\% & 5.8 $\pm$ 2\% & 9.3x &220 $\pm$ 5\% & 27.9 $\pm$ 0.5\% & 7.9x \\\hline
KeyedDB 256 & 25 $\pm$ 4\% & 2.5 $\pm$ 1\% & 10x &110 $\pm$ 7\% & 7.3 $\pm$ 2\% & 15x &420 $\pm$ 4\% & 33.5 $\pm$ 0.6\% & 13x \\\hline
MatrixMult8x8  & 42 $\pm$ 3\% & 41 $\pm$ 4\% & 1.0x &94 $\pm$ 4\% & 79 $\pm$ 3\% & 1.2x &300 $\pm$ 10\% & 310 $\pm$ 1\% & 0.97x \\\hline
Edit Distance 128 & 18 $\pm$ 3\% & 18 $\pm$ 3\% & 1.0x &40 $\pm$ 8\% & 40 $\pm$ 6\% & 1.0x &120 $\pm$ 9\% & 150 $\pm$ 3\% & 0.8x \\\hline
Millionaires 8192& 13 $\pm$ 4\% & 3.2 $\pm$ 1\% & 4.1x &52 $\pm$ 3\% & 8.5 $\pm$ 2\% & 6.1x &220 $\pm$ 5\% & 38.4 $\pm$ 0.9\% & 5.7x \\\hline
\multicolumn{10}{c}{}
\end{tabular}
\caption{Timing results from outsourcing the garbled circuit evaluation from a single server process. Results in seconds.}
\label{table:toserver2}
\end{table*}

\noindent{\bf Outsourcing to a Server Process}

PartialGC can be used in other scenarios than just outsourcing to a mobile
device. It can outsource garbled
circuit evaluation from a single server process and retain performance
benefits over a single server process of CMTB. For this experiment the outsourcing party has a single
thread. Table~\ref{table:toserver2} displays these results and shows that in
the KeyedDB 256 program, PartialGC has a 92\% speedup over CMTB.
As with the outsourced mobile case, keyed database problems perform
particularly well in PartialGC. Because the computationally-intensive input
consistency check is a greater bottleneck on mobile devices than servers,
these improvements for most programs are less dramatic. In particular, both
edit distance and matrix multiplication programs benefit from higher
computational power and their bottlenecks on a server are no longer input
consistency; as a result, they execute faster in CMTB than in
PartialGC.

\subsection{Bandwidth}

\begin{table}[t]
\centering
\small
\begin{tabular}{ |c | r | r| r| c| c| c| c| c| c| c| c| c|}
\hline
  & \multicolumn{3}{|c|}{ 256 Circuits}  \\
\hline
  & CMTB & PartialGC &  \\ \hline

KeyedDB 64  &64992308 & 3590416 & 18x \\\hline
KeyedDB 128   &129744948 & 3590416 & 36x \\\hline
KeyedDB 256   &259250228 & 3590416 & 72x \\\hline
MatrixMult8x8  &71238860 & 35027980 & 2.0x \\\hline
Edit Distance 128  &2615651 & 4108045 & 0.64x \\\hline
Millionaires 8192 &155377267 & 67071757 & 2.3x \\\hline
\multicolumn{4}{c}{}

\end{tabular}
\caption{Bandwidth comparison of CMTB and PartialGC. Bandwidth counted by instrumenting  PartialGC to count the bytes it was sending and receiving and then adding them together. Results in bytes.}
\label{table:keyeddbphoneband2}
\end{table}

Since the main reason for outsourcing a computation is to save on resources,
we give results showing a decrease in the evaluator's bandwidth. Bandwidth
is counted by making the evaluator to count the number of bytes
PartialGC sends and receives to either server. Our best result gives a 98\%
reduction in bandwidth (see Table~\ref{table:keyeddbphoneband2}). For the
edit distance, the extra bandwidth used in the outsourced oblivious transfer
for all circuits, instead of only the evaluation circuits, exceeds any
benefit we would otherwise have received.

\subsection{Secure Friend Finder}
\label{sec:app}

Many privacy-preserving applications can benefit from using PartialGC to cache
values for state. As a case study, we developed a privacy-preserving friend
finder application, where users can locate nearby friends without any user
divulging their exact location. In this application, many different mobile phone
clients use a consistent generator (a server application) and outsource
computation to a cloud. The generator must be the same for all computations; the
cloud must be the same for each computation. The cloud and generator are two
different parties. After each computation, the map is updated when PartialGC
saves the current state of the map as wire labels. Without PartialGC outsourcing
values to the cloud, the wire labels would have to be transferred directly
between mobile devices, making a multi-user application difficult or impossible.

We define three privacy-preserving operations that comprise the
application's functionality:

\noindent {\bf MapStart} - The three parties (generator, evaluator, cloud)
create a ``blank'' map region, where all locations in the map are blank and
remain that way until some mobile party sets a location to his or her ID.

\noindent {\bf MapSet} - The mobile party sets a single map cell to a new
value. This program takes in partial values from the generator and cloud and
outputs a location selected by the mobile party.

\noindent {\bf MapGet} - The mobile party retrieves the contents of a single map
cell. This program retrieves partial values from the generator and cloud and
outputs any ID set for that cell to the mobile.

In the application, each user using the {\it Secure Friend Finder} has a
unique ID that represents them on the map. We divide the map into ``cells'',
where each cell is a set amount of area. When the user presses ``Set New
Location,'' the program will first look to determine if that cell is
occupied. If the cell is occupied, the user is informed he is near a friend.
Otherwise the cell is updated to contain his user ID and remove his ID
from his previous location. We assume a maximum of 255 friends in our
application since each cell in the map is 8 bits. 

\begin{figure}[t]
\centering
\includegraphics[width=3in]{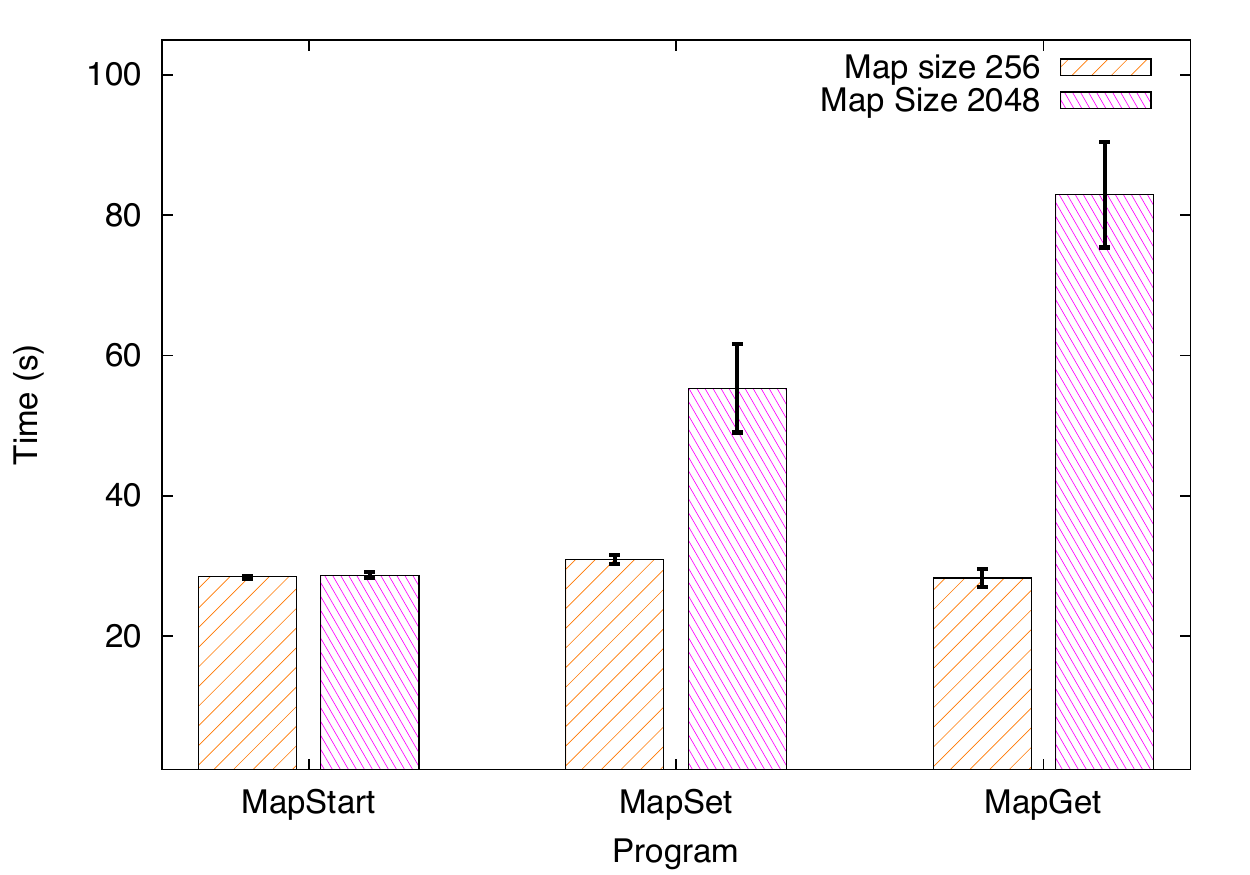}
\caption{Run time comparison of our map programs with two different map sizes.}
\label{fig:mapprogs}
\end{figure}

\begin{figure}[t]
  \centering
  \subfloat[Location selected.]{
\includegraphics[width=0.22\textwidth]{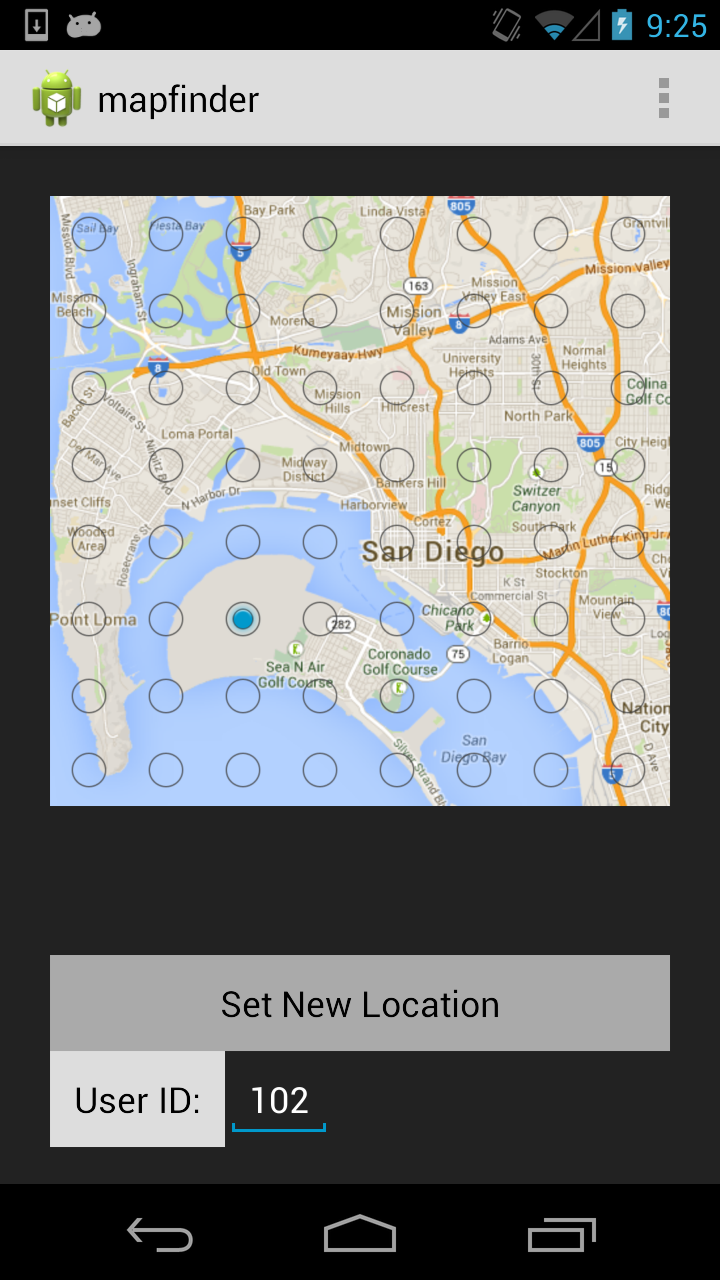}
\label{fig:dij20}
}
\subfloat[After computation.]{
\includegraphics[width=0.22\textwidth]{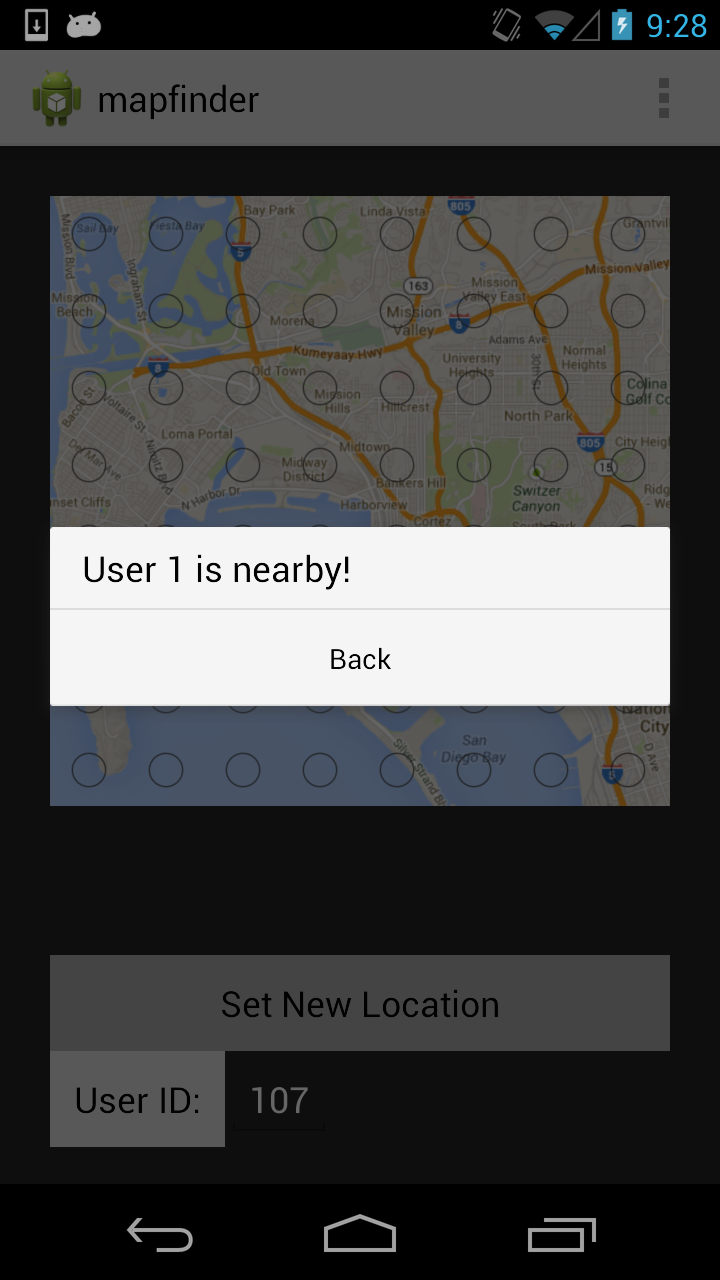}
\label{fig:dij50}
}
\caption{Screenshots from our application. (a) shows the map
with radio buttons a user can select to indicate position. (b) show the result after ``set new position" is pressed when a user is present. The application is set to use 64 different map locations. Map image from Google Maps.} 
\label{fig:maps}
\end{figure}

Figure~\ref{fig:mapprogs} shows the performance of these programs in the
malicious model with a $2^{-80}$ security parameter (evaluated over 256
circuits). We consider map regions containing both 256 and 2048 cells.
For maps of 256 cells, each operation takes
about 30 seconds.\footnote{Our 64-cell map, as seen in figure~\ref{fig:mapprogs}, also takes about 30 seconds for each operation.} 
As there are three operations for each ``Set New
Location'' event, the total execution time is about 90 seconds, while
execution time for 2048 cells is about 3 minutes. The
bottleneck of the 64 and 256 cell maps is the outsourced oblivious transfer,
which is not affected by the number of cells in the map. The vastly larger
circuit associated with the 2048-cell map makes getting and setting values
slower operations, but these results show such an application is practical
for many scenarios. 

\noindent {\bf Example} - As an example, two friends initiate a friend finder computation using Amazon as the cloud and Facebook as the generator. The first friend goes out for a coffee at a caf\'{e}. The second friend, riding his bike, gets a message that his friend is nearby and looks for a few minutes and finds him in the caf\'{e}. Using this application prevents either Amazon or Facebook from knowing either user's location while they are able to learn whether they are nearby.

\subsection{Discussion}
\label{sub:dis}

{\noindent\bf Analysis of improvements}

We analyzed our results and found the improvements came from three places: the improved sS13 consistency check, the saving and reusing of values, and the fixed oblivious transfer. In the case of the sS13 consistency check, there are two reasons for the improvement: first, there is less network traffic, and second, it uses symmetric key operations instead of exponentiations. In the case of saving and reusing values, we save time with the faster input consistency check and by not requiring a user to recompute a circuit multiple times. Lastly, we reduced the runtime and bandwidth by fixing parts of the OOT. The previous outsourced oblivious transfer performed the primitive OT $S$ ($S$ being the number of circuits) times instead of a single time, which turn forced many extra exponentiations. Each amount of improvement varies depending upon the circuit. 

{\noindent\bf Output check}


Although the garbled circuit is larger for our output check, this check performs less cryptographic operations for the outsourcing party, as the evaluator only has to perform a MAC on the output of the garbled circuit. We use this check to demonstrate using a MAC can be an efficient output check for a low power device when the computational power is not equivalent across all parties. 

{\noindent\bf Commit Cut-and-Choose vs OT Cut-and-Choose}

Our results unexpectedly showed that the sS13 OT cut-and-choose used in
PartialGC is actually slower than the KSS12 commit cut-and-choose used in CMTB
in our experimental setup. Theoretically, sS13, which requires fewer
cryptographic operations, as it generates the garbled circuit only once, should
be the faster protocol. The difference between the two cut-and-choose protocols
is the network usage -- instead of $\frac{2}{5}$ of the circuits (CMTB), {\it all} 
the circuits must be transmitted in sS13. The sS13 cut-and-choose is
required in our protocol so that the cloud can check that the generator creates
the correct gates.

\begin{techreport}
\subsection{Implementation Optimizations}

We proceeded to optimize our system in light of the slowdown we saw when compared to CMTB for circuits with large amounts of gates. We made the following changes: (1) turn AES-NI on, as it was not turned on by default in CMTB (or \cite{Kreuter2012}, which CMTB is based on), (2) hand-optimize the garbled gate generation and evaluation to remove excess memory operations, (3) remove the need for network I/O for XOR gates from the underlying implementation (previously, 4 bytes were spuriously transmitted for each XOR gate), (4) batch process gates to reduce the overhead of networking for each gate, and (5) remove unnecessary hash calls that existed in PartialGC as an artifact of being built on CMTB.

\subsection{Corrections of Underlying Implementation}

We made two corrections to the implementation of PartialGC that are artifacts of the underlying implementations. The first error was from KSS and while the other was from CMTB. We performed the following changes: (1) further correct the OT phase of CMTB and (2) add a missing input encoding phase that was supposed to exist in KSS. The first error was straightforward: rather than performing a single set of OTs and then extending it to all circuits, after the single set of OTs in CMTB, a matrix transformation was performed for each circuit (instead of a single matrix transformation). We removed this error and added the necessary correction, i.e. after the single set of OTs were performed, the results from the OTs were extended in the same manner as in our protocol description. To correct the second error, we added the missing input encoding step for the evaluator's input. Note that KSS and all subsequent systems built from it do not have this input encoding. Without the input encoding, a selective failure attack can be performed easily by the generator in order to gain information about a single bit of the evaluator`s input. 

\subsection{Results from Correct and More Optimal Implementation}

\begin{table*}[t]
\centering\small
  \begin{tabular}{| c | r |  r | r | r |  r |  r | r |  r |  r|}
\hline
  & \multicolumn{3}{|c|}{ 16 Circuits}  & \multicolumn{3}{|c|}{ 64 Circuits}  & \multicolumn{3}{|c|}{ 256 Circuits}  \\
\hline
 & \multicolumn{1}{|c|}{Imp.} & \multicolumn{1}{|c|}{Orig.} &  & \multicolumn{1}{|c|}{Imp.} & \multicolumn{1}{|c|}{Orig.} &  & \multicolumn{1}{|c|}{Imp.} & \multicolumn{1}{|c|}{Orig.}  &  \\ \hline

KeyedDB 64 & 4 $\pm$ 10\% &3.5 $\pm$ 3\% & 0.92x&4.4 $\pm$ 5\% &8.3 $\pm$ 5\% & 1.9x&7.6 $\pm$ 6\% &26 $\pm$ 2\% & 3.4x\\[1pt]\hline 
KeyedDB 128 & 3.8 $\pm$ 10\% &4.4 $\pm$ 8\% & 1.1x&4.5 $\pm$ 8\% &9.5 $\pm$ 4\% & 2.1x&8.1 $\pm$ 4\% &31 $\pm$ 3\% & 3.8x\\[1pt]\hline 
KeyedDB 256 & 4.0 $\pm$ \phantom{0}4\% &4.6 $\pm$ 2\% & 1.1x&4.7 $\pm$ 9\% &12 $\pm$ 6\% & 2.7x&9.3 $\pm$ 4\% &38 $\pm$ 5\% & 4.0x\\[1pt]\hline 
MatrixMult8x8 & 21 $\pm$ \phantom{0}2\% &46 $\pm$ 4\% & 2.2x&29 $\pm$ 4\% &100 $\pm$ 7\% & 3.5x&69 $\pm$ 2\% &370 $\pm$ 5\% & 5.4x\\[1pt]\hline 
EditDist 128 & 7.8 $\pm$ \phantom{0}4\% &22 $\pm$ 3\% & 2.8x&10 $\pm$ 4\% &50 $\pm$ 9\% & 4.8x&21 $\pm$ 2\% &180 $\pm$ 6\% & 8.9x\\[1pt]\hline 
Millionaires 8192 & 24 $\pm$ \phantom{0}5\% &7.3 $\pm$ 6\% & 0.30x&30 $\pm$ 3\% &20 $\pm$ 2\% & 0.68x&78 $\pm$ 3\% &70 $\pm$ 2\% & 0.89x\\[1pt]\hline 
\multicolumn{10}{c}{}
  \end{tabular}
  \caption{Comparing the original PartialGC and the improved version of PartialGC. Results in seconds.}
  \label{tab:opexecution}
\end{table*}

In Table~\ref{tab:opexecution} we present results from the corrected and more optimal implementation of PartialGC. We observe the following:

\begin{enumerate}
\item The program that has a large evaluator's input and very little gates is slightly slower due to the fixed OT error and added input encoding (Millionaires).

\item The program with a large circuit size when compared with the input sizes of both the generator and evaluator has improved runtime performance (Edit distance).

\item The program we tested that has high input and also has a high gate count is improved (Matrix Mult).

\item The program that relies mostly on the generator's input size with a low amount of gates is largely unaffected by the OT change or the added input encoding but is still improved by the optimizations to the garbled gate runtime (Keyed DB).
\end{enumerate}

\end{techreport}

\begin{techreport}

\subsection{ SFE Engineering Insights}

Given our experience from building on other frameworks, we provide our insights:

\vspace{.5mm}

\setlist{nolistsep=0,leftmargin=3mm}
\begin{enumerate}
\item If runtime results do not describe the intuition of the protocol then there is most likely something incorrect in the implementation. For instance, if the average time to evaluate garbled gates is greater than the average time to generate the garbled gates there is most likely a problem in the garbled circuit evaluation phase.
\item Although comparing the time of garbling and evaluating can be interesting in its own right, evaluating the total time of full garbled circuit garbling and evaluation (including network overhead) is also insightful as networking and related operations can be the bottleneck in a practical system. This includes network usage, the effects of a cut-and-choose protocol, and the time it takes to get the next gate from the interpreter or circuit file.
\item When using another implementation, check to verify each protocol step exists in the implementation.
\item Implementing checks at the circuit layer that are exposed to an end-user is not worth the time saved by not encoding them directly into the garbled circuit. This comes from our experience with our output consistency check, which was difficult to create correctly for each test program.
\item Ensure that all the features of a developed compiler and execution system are thoroughly unit tested.

\end{enumerate}
\end{techreport}

\section{Related Work}
\label{sec:relwork}

SFE was first described by Yao in his seminal paper~\cite{Yao1982} on the
subject. The first general purpose platform for SFE,
Fairplay~\cite{Malkhi2004}, was created in 2004. Fairplay had both a compiler
for creating garbled circuits, and a run-time system for executing them.
Computations involving three or more parties have also been examined; one of the
earliest examples is FairplayMP~\cite{Ben-David2008}. There have been multiple
other implementations since, in both
semi-honest~\cite{Burkhart2010,Damgard2009,Henecka2010,Holzer2012,Zhang2013} and
malicious settings~\cite{Kreuter2013,Shelat2011}.

Optimizations for garbled circuits include the free-XOR
technique~\cite{Kolesnikov2008}, garbled row reduction~\cite{Pinkas2009},
rewriting computations to minimize SFE~\cite{Ker13}, and
pipelining~\cite{Huang2011}. Pipelining allows the evaluator to proceed with the
computation while the generator is creating gates.

Kreuter {\it et al.}~\cite{Kreuter2012} included both an optimizing compiler and an efficient
run-time system using a parallelized implementation of SFE in the malicious
model from~\cite{Shelat2011}.

The creation of circuits for SFE in a fast and efficient manner is one of
the central problems in the area. Previous compilers, from Fairplay to
KSS12, were based on the concept of creating a complete circuit and then
optimizing it. PAL~\cite{Mood2012} improved such systems by
using a simple template circuit, reducing memory usage by 
orders of magnitude. PCF~\cite{Kreuter2013} built from this and used a more
advanced representation to reduce the disk space used.

Other methods for performing MPC involve homomorphic encryption~\cite{Bendlin2011,Gentry2012}, secret sharing~\cite{BLW08}, and ordered binary decision diagrams~\cite{Kruger2006}.  A general privacy-preserving computation protocol that uses homomorphic encryption and was designed specifically for mobile devices can be found in~\cite{EMOC}. There are also custom protocols designed for particular privacy-preserving computations; for example, Kamara et al.~\cite{MSR201363} showed how to scale server-aided Private Set Intersection to billion-element sets with a custom protocol.

Previous reusable garbled-circuit schemes include that of Brand\~{a}o~\cite{Bran2013}, which uses homomorphic encryption,  Gentry {\it et al.}~\cite{Gentry2013}, which uses attribute-based functional encryption, and Goldwasser {\it et al.}~\cite{Goldwasser2013}, which introduces a succinct functional encryption scheme. These previous works are purely theoretical; none of them provides experimental performance analysis.  There is also recent theoretical work on reusing encrypted garbled-circuit values~\cite{LO13,GHLORW14,LO14} in the ORAM model; it uses a variety of techniques, including garbled circuits and identity-based encryption, to execute the underlying low-level operations (program state, read/write queries, etc.). Our scheme for reusing encrypted values is based on completely different techniques; it enables us to do new kinds of computations, thus expanding the set of things that can be computed using garbled circuits.

The Quid-Pro-Quo-tocols system~\cite{HKE12} allows fast execution with a
single bit of leakage. The garbled circuit is executed
twice, with the parties switching roles in the latter execution, then 
running a secure protocol to ensure that the output from both
executions are equivalent; if this fails, a single bit may be leaked due to
the selective failure attack.

\section{Conclusion}
\label{sec:conc}

This paper presents PartialGC, a server-aided SFE scheme allowing the
reuse of encrypted values to save the costs of input validation and to allow
for the saving of state, such that the costs of multiple computations may be
amortized. Compared to the server-aided outsourcing scheme by CMTB, we
reduce costs of computation by up to 96\% and bandwidth costs by up to 98\%.
Future work will consider the generality of the encryption re-use
scheme to other SFE evaluation systems and large-scale systems
problems that benefit from the addition of state, which can open up new and
intriguing ways of bringing SFE into the practical realm.

\paragraph{Acknowledgements}
This
material is based on research sponsored by the Defense Advanced Research
Projects Agency (DARPA) and the Air Force Research Laboratory under
contracts
FA8750-11-2-0211 and FA8750-13-2-0058. The U.S. Government is authorized to reproduce and distribute reprints for Governmental purposes notwithstanding any copyright
notation thereon. The views and conclusions contained herein are those of
the authors and should not be interpreted as necessarily representing the
official policies or endorsements, either expressed or implied, of DARPA
or the U.S. Government.

\small

\bibliographystyle{abbrv}
\bibliography{arxiv}

\begin{thebibliography}{10}

\bibitem{Bellare1990}
M.~Bellare and S.~Micali.
\newblock {Non-Interactive Oblivious Transfer and Applications}.
\newblock In {\em {Proceedings of Advances in Cryptology (CRYPTO' 89)}}, 1990.

\bibitem{Ben-David2008}
A.~Ben-David, N.~Nisan, and B.~Pinkas.
\newblock {FairplayMP: A System for Secure Multi-Party Computation}.
\newblock In {\em Proceedings of the ACM conference on Computer and
  Communications Security (CCS'08)}, 2008.

\bibitem{Bendlin2011}
R.~Bendlin, I.~Damg{\aa}rd, C.~Orlandi, and S.~Zakarias.
\newblock {Semi-Homomorphic Encryption and Multiparty Computation}.
\newblock In {\em Proceedings of the 30th Annual International Conference on
  Theory and Applications of Cryptographic Techniques: Advances in Cryptology
  (EUROCRYPT'11)}, 2011.

\bibitem{BLW08}
D.~Bogdanov, S.~Laur, and J.~Willemson.
\newblock {Sharemind: A Framework for Fast Privacy-Preserving Computations}.
\newblock In {\em Proceedings of the 13th European Symposium on Research in
  Computer Security (ESORICS'08)}, 2008.

\bibitem{Bran2013}
L.~T. A.~N. Brand\~{a}o.
\newblock {Secure Two-Party Computation with Reusable Bit-Commitments, via a
  Cut-and-Choose with Forge-and-Lose Technique}.
\newblock Technical report, University of Lisbon, 2013.

\bibitem{Burkhart2010}
M.~Burkhart, M.~Strasser, D.~Many, and X.~Dimitropoulos.
\newblock {SEPIA: Privacy-preserving Aggregation of Multi-domain Network Events
  and Statistics}.
\newblock In {\em {Proceedings of the 19th USENIX Conference on Security
  (SECURITY'10)}}, 2010.

\bibitem{EMOC}
H.~Carter, C.~Amrutkar, I.~Dacosta, and P.~Traynor.
\newblock {For Your Phone Only: Custom Protocols For Efficient Secure Function
  Evaluation On Mobile Devices}.
\newblock {\em {Journal of Security and Communication Networks (SCN)}},
  7(7):1165--1176, 2014.

\bibitem{CMTB2013}
H.~Carter, B.~Mood, P.~Traynor, and K.~Butler.
\newblock {Secure Outsourced Garbled Circuit Evaluation for Mobile Devices}.
\newblock In {\em {Proceedings of the USENIX Security Symposium
  (SECURITY'13)}}, 2013.

\bibitem{Damgard2009}
I.~Damg{\aa}rd, M.~Geisler, M.~Kr{\o}igaard, and J.~B. Nielsen.
\newblock {Asynchronous Multiparty Computation: Theory and Implementation}.
\newblock In {\em {Proceedings of the 12th International Conference on Practice
  and Theory in Public Key Cryptography (PKC '09)}}, 2009.

\bibitem{Gentry2013}
C.~Gentry, S.~Gorbunov, S.~Halevi, V.~Vaikuntanathan, and D.~Vinayagamurthy.
\newblock {How to Compress (Reusable) Garbled Circuits}.
\newblock Cryptology ePrint Archive, Report 2013/687, 2013.
\newblock \url{http://eprint.iacr.org/}.

\bibitem{GHLORW14}
C.~Gentry, S.~Halevi, S.~Lu, R.~Ostrovsky, M.~Raykova, and D.~Wichs.
\newblock {Garbled RAM Revisited}.
\newblock In {\em Advances in Cryptology (EUROCRYPT'14)}, 2014.

\bibitem{Gentry2012}
C.~Gentry, S.~Halevi, and N.~P. Smart.
\newblock {Homomorphic Evaluation of the AES Circuit}.
\newblock In {\em {Proceedings of Advances in Cryptology (CRYPTO'12)}}, 2012.

\bibitem{Goldwasser2013}
S.~Goldwasser, Y.~Kalai, R.~A. Popa, V.~Vaikuntanathan, and N.~Zeldovich.
\newblock {Reusable Garbled Circuits and Succinct Functional Encryption}.
\newblock In {\em Proceedings of the ACM Symposium on Theory of Computing
  (STOC)}, 2013.

\bibitem{Goyal2008}
V.~Goyal, P.~Mohassel, and A.~Smith.
\newblock {Efficient Two Party and Multi Party Computation Against Covert
  Adversaries}.
\newblock In {\em Advances in Cryptology (EUROCRYPT'08)}, 2008.

\bibitem{Halevi11}
S.~Halevi, Y.~Lindell, and B.~Pinkas.
\newblock {Secure Computation on the Web: Computing Without Simultaneous
  Interaction}.
\newblock In {\em {Advances in Cryptology (CRYPTO'11}}, 2011.

\bibitem{Henecka2010}
W.~Henecka, S.~K\"{o}gl, A.-R. Sadeghi, T.~Schneider, and I.~Wehrenberg.
\newblock {TASTY: Tool for Automating Secure Two-Party Computations}.
\newblock In {\em Proceedings of the ACM conference on Computer and
  Communications Security (CCS'10)}, 2010.

\bibitem{Holzer2012}
A.~Holzer, M.~Franz, S.~Katzenbeisser, and H.~Veith.
\newblock {Secure Two-party Computations in ANSI C}.
\newblock In {\em Proceedings of the 2012 ACM Conference on Computer and
  Communications Security (CCS '12)}, 2012.

\bibitem{Huang2011}
Y.~Huang, D.~Evans, J.~Katz, and L.~Malka.
\newblock {Faster Secure Two-Party Computation Using Garbled Circuits}.
\newblock In {\em {Proceedings of the USENIX Security Symposium (SECURITY
  '11)}}, 2011.

\bibitem{HKE12}
Y.~Huang, J.~Katz, and D.~Evans.
\newblock {Quid-Pro-Quo-tocols: Strengthening Semi-Honest Protocols with Dual
  Execution}.
\newblock {\em IEEE Symposium on Security and Privacy (OAKLAND'12)}, 2012.

\bibitem{Ishai2003}
Y.~Ishai, J.~Kilian, K.~Nissim, and E.~Petrank.
\newblock {Extending Oblivious Transfers Efficiently}.
\newblock In {\em Advances in Cryptology (CRYPTO'03)}, 2003.

\bibitem{MSR201363}
S.~Kamara, P.~Mohassel, M.~Raykova, and S.~Sadeghian.
\newblock {Scaling Private Set Intersection to Billion-Element Sets}.
\newblock Technical Report MSR-TR-2013-63, Microsoft Research, 2013.

\bibitem{Kamara2012}
S.~Kamara, P.~Mohassel, and B.~Riva.
\newblock {Salus: A System for Server-Aided Secure Function Evaluation}.
\newblock In {\em {Proceedings of the ACM conference on Computer and
  Communications Security (CCS'12)}}, 2012.

\bibitem{Ker13}
F.~Kerschbaum.
\newblock {Expression Rewriting for Optimizing Secure Computation}.
\newblock In {\em {Conference on Data and Application Security and Privacy
  (CODASPY'13)}}, 2013.

\bibitem{Kiraz2006}
M.~S. Kiraz and B.~Schoenmakers.
\newblock {A Protocol Issue for the Malicious Case of Yao's Garbled Circuit
  Construction}.
\newblock In {\em {Proceedings of Symposium on Information Theory in the
  Benelux}}, 2006.

\bibitem{Kolesnikov2008}
V.~Kolesnikov and T.~Schneider.
\newblock {Improved Garbled Circuit: Free XOR Gates and Applications}.
\newblock In {\em Proceedings of the International Colloquium on Automata,
  Languages and Programming (ICALP'08)}, 2008.

\bibitem{Kreuter2013}
B.~Kreuter, B.~Mood, a.~shelat, and K.~Butler.
\newblock {PCF: A Portable Circuit Format for Scalable Two-Party Secure
  Computation}.
\newblock In {\em {Proceedings of the USENIX Security Symposium (SECURITY
  '13)}}, 2013.

\bibitem{Kreuter2012}
B.~Kreuter, a.~shelat, and C.-H. Shen.
\newblock Billion-gate secure computation with malicious adversaries.
\newblock In {\em {Proceedings of the USENIX Security Symposium
  (SECURITY'12)}}, 2012.

\bibitem{Kruger2006}
L.~Kruger, S.~Jha, E.-J. Goh, and D.~Boneh.
\newblock {Secure Function Evaluation with Ordered Binary Decision Diagrams}.
\newblock In {\em {Proceedings of the ACM Conference on Computer and
  Communications Security (CCS)}}, 2006.

\bibitem{Lindell2007}
Y.~Lindell and B.~Pinkas.
\newblock {An Efficient Protocol for Secure Two-Party Computation in the
  Presence of Malicious Adversaries}.
\newblock In {\em Advances in Cryptology (EUROCRYPT'07)}, 2007.

\bibitem{LO13}
S.~Lu and R.~Ostrovsky.
\newblock {How to Garble RAM Programs}.
\newblock In {\em Advances in Cryptology (EUROCRYPT'13)}. 2013.

\bibitem{LO14}
S.~Lu and R.~Ostrovsky.
\newblock Garbled ram revisited, part ii.
\newblock Cryptology ePrint Archive, Report 2014/083, 2014.
\newblock \url{http://eprint.iacr.org/}.

\bibitem{Malkhi2004}
D.~Malkhi, N.~Nisan, B.~Pinkas, and Y.~Sella.
\newblock {Fairplay--A Secure Two-Party Computation System}.
\newblock In {\em {Proceedings of the USENIX Security Symposium
  (SECURITY'04)}}, 2004.

\bibitem{Mood2012}
B.~Mood, L.~Letaw, and K.~Butler.
\newblock {Memory-Efficient Garbled Circuit Generation for Mobile Devices}.
\newblock In {\em Proceedings of the IFCA International Conference on Financial
  Cryptography and Data Security (FC'12)}, 2012.

\bibitem{Naor1999a}
M.~Naor and B.~Pinkas.
\newblock {Oblivious Transfer and Polynomial Evaluation}.
\newblock In {\em {Proceedings of the Annual ACM Symposium on Theory of
  Computing (STOC'99)}}, 1999.

\bibitem{Naor2001}
M.~Naor and B.~Pinkas.
\newblock {Efficient Oblivious Transfer Protocols}.
\newblock In {\em {Proceedings of the Annual ACM-SIAM Symposium on Discrete
  algorithms (SODA'01)}}, 2001.

\bibitem{Pinkas2009}
B.~Pinkas, T.~Schneider, N.~P. Smart, and S.~C. Williams.
\newblock {Secure Two-Party Computation is Practical}.
\newblock In {\em Advances in Cryptology (ASIACRYPT'09)}, 2009.

\bibitem{Shelat2011}
a.~shelat and C.-H. Shen.
\newblock {Two-Output Secure Computation with Malicious Adversaries}.
\newblock In {\em Advances in Cryptology (EUROCRYPT'11)}, 2011.

\bibitem{shelat13}
a.~shelat and C.-H. Shen.
\newblock {Fast Two-Party Secure Computation with Minimal Assumptions}.
\newblock In {\em {Conference on Computer and Communications Security
  (CCS'13)}}, 2013.

\bibitem{Yao1982}
A.~C. Yao.
\newblock {Protocols for secure computations}.
\newblock In {\em Proceedings of the IEEE Symposium on Foundations of Computer
  Science (FOCS'82)}, 1982.

\bibitem{Zhang2013}
Y.~Zhang, A.~Steele, and M.~Blanton.
\newblock {PICCO: A General-purpose Compiler for Private Distributed
  Computation}.
\newblock In {\em {Proceedings of the ACM Conference on Computer Communications
  Security (CCS'13)}}, 2013.

\end{thebibliography}

\normalsize
\appendix
\section{CMTB Protocol}\label{appendix:CMTB}

As we are building off of the CMTB garbled circuit execution system, we give
an abbreviated version of the protocol. In our description we refer to the
generator, the cloud, and the evaluator. The cloud is the party the
evaluator outsources her computation to.

\paragraph{\bf Circuit generation and check}  
The template for the garbled circuit is augmented to add one-time XOR pads on the
output bits and split the evaluator's input wires per the input encoding
scheme. The generator generates the necessary garbled circuits and
commits to them and sends the commitments to the evaluator. The generator
then commits to input labels for the evaluator's inputs. 

CMTB relies on Goyal et al.'s~\cite{Goyal2008} random seed check, which was
implemented by Kreuter et al.~\cite{Kreuter2012} to combat generation of
incorrect circuits. This technique uses a cut-and-choose style protocol to
determine whether the generator created the correct circuits by creating
and committing to many different circuits. Some of those circuits are used for
evaluation, while the others are used as check circuits.

\paragraph{\bf Evaluator's inputs} 
Rather than a two-party oblivious transfer, we
perform a three-party {\em outsourced oblivious transfer}.
An outsourced oblivious transfer is an OT that gets the
select bits from one party, the wire labels from another,
and returns the selected wire labels to a third party.
The party that selects the wire labels does not learn
what the wire labels are, and the party that inputs the 
wire labels does not learn which wire was selected; 
the third party only learns the selected wire labels. 
In CMTB, the generator offers up wire labels, the evaluator 
provides the select bits, and the cloud receives the selected labels.
CMTB uses the Ishai OT extension~\cite{Ishai2003} to reduce 
the number of OTs.

CMTB uses an encoding technique from Lindell and Pinkas \cite{Lindell2007}, which prevents the generator from finding out any
information about the evaluator's input if a selective failure attack
transpires. CMTB also uses the commitment technique of Kreuter et
al.~\cite{Kreuter2012} to prevent the generator from swapping the two
possible outputs of the oblivious transfer. To ensure the evaluator's input is consistent across all circuits, CMTB uses a technique from Lindell and Pinkas~\cite{Lindell2007}, whereby the inputs are derived from a single oblivious transfer.

\paragraph{\bf Generator's input and consistency check} 
The generator sends his input to the cloud for the evaluation circuits. Then
the generator, evaluator, and cloud all work together to prove the input
consistency of the generator's input.  For the generator's input consistency
check, CMTB uses the malleable-claw free construction from shelat and
Shen~\cite{Shelat2011}.  

\paragraph{\bf Circuit evaluations} 
The cloud evaluates the garbled circuits marked for evaluation and checks
the circuits marked for checking. The cloud enters in the generator and
evaluator's input into each garbled circuit and evaluates each circuit. The
output for any particular bit is then the majority output between all
evaluator circuits. The cloud then recreates each check circuit.  The cloud
creates the hashes of each garbled circuit and sends those hashes to the
evaluator. The evaluator then verifies the hashes are the same as the ones
the generator previously committed to.

\paragraph{\bf Output consistency check and output} 
The three parties prove together that the cloud did not modify the output
before she sent it to the generator or evaluator.  Both the evaluator and
generator receive their respective outputs. All outputs are blinded by
the respective party's one-time pad inside the garbled circuit to prevent
the cloud from learning what any output bit represents.

CMTB uses the XOR one-time pad technique from Kiraz~\cite{Kiraz2006} to
prevent the evaluator from learning the generator's real output. To prevent
output modification, CMTB uses the
witness-indistinguishable zero-knowledge proof from Kreuter et
al.~\cite{Kreuter2012}.

\section{Overhead of Reusing Values}\label{appendix:reuse}

\begin{figure}[t]
\centering
\includegraphics[width=3in]{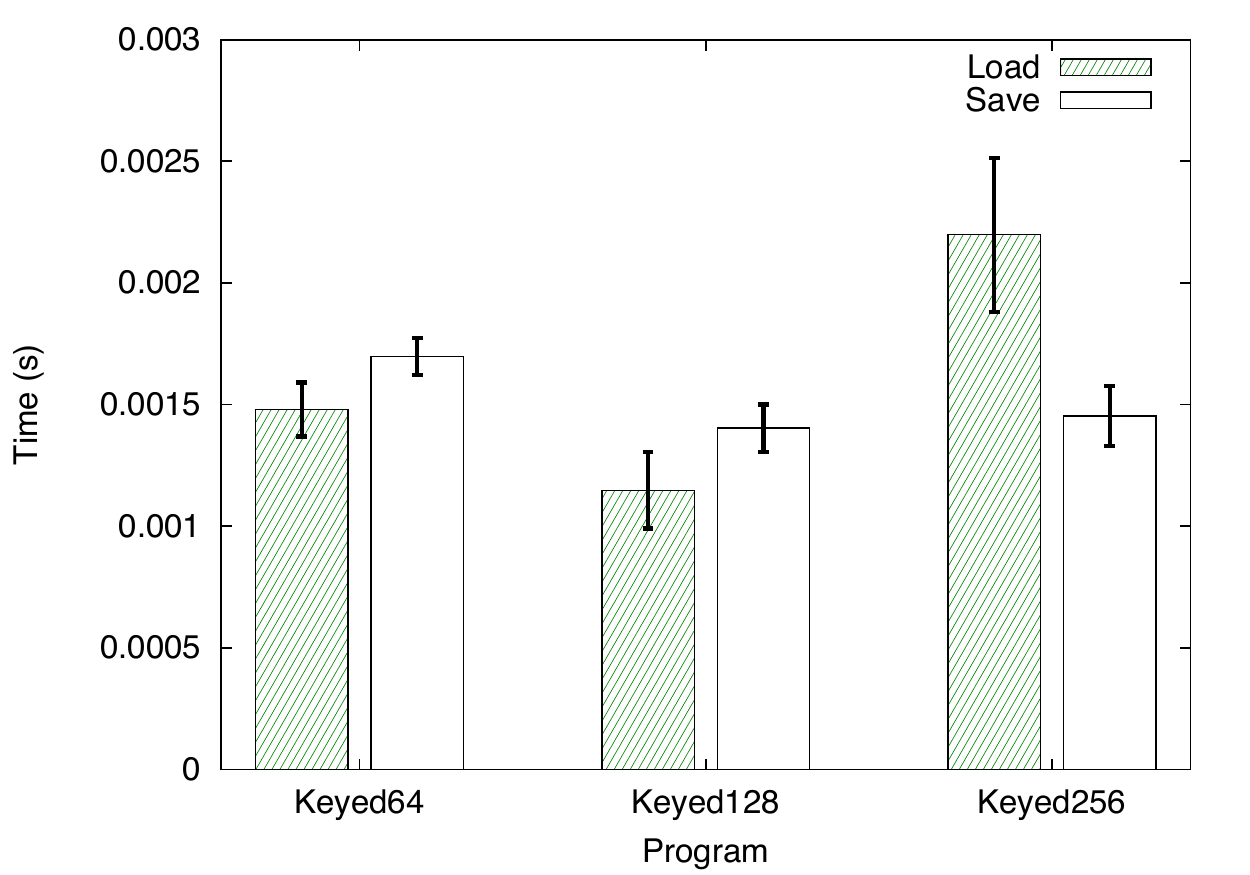}
\caption{The amount of time it takes to save and load a bit in PartialGC when using 256 circuits.}
\label{fig:saveload}
\end{figure}

We created several versions of the keyed database program to determine the runtime of saving and loading the database on a per bit basis using our system (See Figure~\ref{fig:saveload}). This figure shows it is possible to save and load a large amount of saved wire labels in a relatively short time. The time to load a wire label is larger than the time to save a value since saving only involves saving the wire label to a file and loading involves reading from a file and creating the partial input gates. Although not shown in the figure, the time to save or load a single bit also increases with the circuit parameter. This is because we need $S$ copies of that bit - one for every circuit.

\section{Example Program}\label{appendix:example}

In this section we describe the execution of an {\it attendance application}. Consider a scenario where the host wants each user to sign in from their phones to keep a log of the guests, but also wants to keep this information secret.

This application has three distinct programs. The first program initializes a counter to a number input by the evaluator. The second program, which is used until the last program is called, takes in a name and increments the counter by one. The last program outputs all names and returns the count of users. For this application, users (specifically, their mobile phones) assume the role of evaluators in the protocol (Section~\ref{sec:protocolchanges}).

First, the host runs the initial program to initialize a database. We cannot execute the second program to add names to the log until this is done, lest we reveal that there is no memory saved ({\it i.e.}, there is no one else present).

\noindent
{\bf Protocol in Brief}: In this first program, the cut-and-choose OT is executed to select the circuit split (the circuits that are for evaluation and generation). Both parties save the decryption keys: the cloud saves the keys attained from the OT and the generator saves both possible keys that could have been selected by the cloud. The evaluator performs the OOT with the other parties to input the initial value into the program. There is no input by the generator so the generator's input check does not execute. There is no partial input so that phase of the protocol is skipped. The garbled circuit to set the initial value is executed;  while there is no output to the generator or evaluator, a partial output is produced: the cloud saves the garbled wire value, which it possesses, and the generator saves both possible wire values (the generator does not know what value the cloud has, and the cloud does not know what the value it has saved actually represents). The cloud also saves the circuit split.

Saved memory after the program execution (when the evaluator inputs 0 as the initial value):

\begin{center}
\begin{tabular}{|c|}
\hline
  Count\\\hline
  0\\\hline\hline
  Saved Guests \\\hline
  \\\hline
\end{tabular}
\end{center}

Guest 1 then enters the building and executes the program, entering his name (``Guest 1'') as input.

\noindent
{\bf Protocol in Brief for Second Program}: In this second program, the cut-and-choose OT is not executed. Instead, both the generator and cloud load the saved decryption key values, hash them, and use those values for the check and evaluation circuit information (instead of attaining new keys through an OT, which would break security). The new keys are saved, and the evaluator then performs the OOT for input. The generator does not have any input in this program so the check for the generator's input is skipped. Since there exists a partial input, the generator loads both possible wire values and creates the partial input gates. The cloud loads the attained values, receives the partial input gates from the generator, and then executes (and checks) the partial input gates to receive the garbled input values. The garbled circuit is then executed and partial output saved as before (although there is more data to save for this program as there is a name present in the database).

After executing the second program the memory is as follows:

\begin{center}
\begin{tabular}{|c|}
\hline
  Count\\\hline
  1\\\hline\hline
  Saved Guests \\\hline
  Guest1 \\\hline
\end{tabular}
\end{center}

Guest 2 then enters the dwelling and runs the program. The execution is similar to the previous one (when Guest 1 entered), except that it's executed by Guest 2's phone.

At this point, the memory is as follows:

\begin{center}
\begin{tabular}{|c|}
\hline
  Count\\\hline
  2\\\hline\hline
   Saved Guests \\\hline
  Guest1 \\\hline
  Guest2 \\\hline
\end{tabular}
\end{center}

Guest 3 then enters the dwelling and executes the program as before. At this point, the memory is as follows:

\begin{center}
\begin{tabular}{|c|}
\hline
  Count\\\hline
  3\\\hline\hline
   Saved Guests \\\hline
  Guest1 \\\hline
  Guest2 \\\hline
  Guest3 \\\hline
\end{tabular}
\end{center}

Finally, the host runs the last program that outputs the count and the guests in the database. In this case the count is $3$ and the guests are $Guest1$, $Guest2$, and $Guest3$.

\end{document}